\documentclass[12pt,draftclsnofoot,onecolumn]{IEEEtran}

\DeclareMathAlphabet\mathbfcal{OMS}{cmsy}{b}{n}

\makeatletter
\newcommand{\vast}{\bBigg@{4}}
\makeatother
\usepackage{color}
\usepackage{slashbox}
\usepackage{multirow}
\usepackage{graphicx}
\usepackage{epstopdf}
\usepackage[cmex10]{amsmath}
\usepackage{amssymb}

\usepackage{multirow}
 \usepackage{algpseudocode}
\usepackage{algorithm}
\usepackage{amsthm}
\theoremstyle{definition}

\theoremstyle{remark}

\newtheorem{prop}{Proposition}

\newtheorem{lem}{Lemma}
\makeatletter
\g@addto@macro\th@remark{\thm@headpunct{\normalfont:}}
\makeatother
\makeatletter
\newcommand{\distas}[1]{\mathbin{\overset{#1}{\kern\z@\sim}}}%
\newsavebox{\mybox}\newsavebox{\mysim}
\newcommand{\distras}[1]{%
  \savebox{\mybox}{\hbox{\kern3pt$\scriptstyle#1$\kern3pt}}%
  \savebox{\mysim}{\hbox{$\sim$}}%
  \mathbin{\overset{#1}{\kern\z@\resizebox{\wd\mybox}{\ht\mysim}{$\sim$}}}%
}
\allowdisplaybreaks

\begin{document}

\title{Green Cognitive Relaying: Opportunistically Switching Between Data Transmission and Energy Harvesting}

\author{Nikolaos~I.~Miridakis, Theodoros~A.~Tsiftsis,~\IEEEmembership{Senior Member,~IEEE},\\ George C. Alexandropoulos,~\IEEEmembership{Senior Member,~IEEE}, and\\ M\'erouane Debbah,~\IEEEmembership{Fellow,~IEEE}
\thanks{N. I. Miridakis is with the Department of Computer Engineering, Piraeus University of Applied Sciences, 12244 Aegaleo, Greece (e-mail: nikozm@unipi.gr).}
\thanks{T. A. Tsiftsis is with the Department of Electrical Engineering, Technological Educational Institute of Central Greece, 35100 Lamia, Greece (e-mail: tsiftsis@teiste.gr).}
\thanks{G. C. Alexandropoulos and M. Debbah are with the Mathematical and Algorithmic Sciences Lab, France Research Center, Huawei Technologies Co. Ltd., 92100 Boulogne-Billancourt, France (e-mails: \{george.alexandropoulos,merouane.debbah\}@huawei.com).}
}

\markboth{}%
{}

\maketitle
\begin{abstract}
Energy efficiency has become an encouragement, and more than this, a requisite for the design of next-generation wireless communications standards. In current work, a dual-hop cognitive (secondary) relaying system is considered, incorporating multiple amplify-and-forward relays, a rather cost-effective solution. First, the secondary relays sense the wireless channel, scanning for a primary network activity, and then convey their reports to a secondary base station (SBS). Afterwards, the SBS, based on these reports and its own estimation, decides cooperatively the presence of primary transmission or not. In the former scenario, all the secondary nodes start to harvest energy from the transmission of primary node(s). In the latter scenario, the system initiates secondary communication via a best relay selection policy. Performance evaluation of this system is thoroughly investigated, by assuming realistic channel conditions, i.e., non-identical link-distances, Rayleigh fading, and outdated channel estimation. The detection and outage probabilities as well as the average harvested energy are derived as new closed-form expressions. In addition, an energy efficiency optimization problem is analytically formulated and solved, while a necessary condition in terms of power consumption minimization for each secondary node is presented. From a green communications standpoint, it turns out that energy harvesting greatly enhances the resources of secondary nodes, especially when primary activity is densely present.
\end{abstract}


\IEEEpeerreviewmaketitle

\section{Introduction}
\IEEEPARstart{R}{ecently}, cognitive radio (CR) has emerged as one of the most promising technologies to resolve the issue of spectrum scarcity, caused by the escalating growth in wireless data traffic of next-generation networks \cite{ref20,ref24}. One of the principal requirements of CR is the effectiveness of spectrum sharing performed by secondary (unlicensed) nodes, which is expected to intelligently mitigate any harmful interference caused to the primary (licensed) network nodes. This requirement is directly related to the accuracy of spectrum sensing/sharing techniques, reflecting the reliable detection of primary transmission. Moreover, to further guarantee a sufficient quality level of primary communication, the transmission power of CR is generally limited, such that its interference onto primary users remains below prescribed tolerable levels. However, this dictated constraint dramatically affects the coverage and/or capacity of the secondary communication. Such a condition can be effectively counteracted with the assistance of wireless relaying transmission. In particular, the rather feasible dual-hop multi-relay communication scheme with best relay selection is of paramount interest due to its enhanced performance gains (e.g., see \cite{ref8}-\cite{refextra3} and references therein).

Building on the aforementioned system deployment, the spectrum sensing process can also be significantly enhanced. By means of the so-called cooperative sensing (e.g., see \cite{refextrraa,ref26}), each secondary node may sense the channel in fixed time periods and then forward its sensing measurement to a central secondary base station (SBS), which acts as a fusion center. The latter entity is responsible for the final decision on a primary transmission occurrence, benefiting from the spatial diversity of several sensing reports. Such a distributed (cooperative) spectrum sensing was shown to deliver much more accurate decisions than local (standalone) sensing regarding the detection of primary transmissions \cite{ref26}.

On another front, driven by the ever increasing economical and environmental (e.g., carbon footprint) costs associated with the operating expenditure of communication networks, energy efficiency (EE) has become an important design consideration in current and forthcoming wireless cognitive infrastructures \cite{ref21}-\cite{ref23}. Taking into account that nodes in wireless communication systems are mostly battery-driven while governed by power-hungry equipment, EE appears to be an essential requirement.

\subsection{Related Work and Motivation}
In cooperative CR systems, two different types of relaying protocols have dominated so far, namely, amplify-and-forward (AF) and decode-and-forward (DF) \cite{dohlerbook}. It is noteworthy that AF outperforms the computational-demanding DF in terms of EE and/or power savings. In \cite{ref27} and \cite{ref28}, the tradeoff between consumed power and performance of the secondary system was investigated. However, in these works, the objective was performance enhancement (in terms of either outage probability \cite{ref27} or throughput \cite{ref28}) not EE and/or minimization of power consumption. In \cite{ref29}-\cite{ref31}, the authors focused on the energy minimization of cognitive relaying networks. Yet, EE and data transmission requirements (e.g., in terms of a data rate and/or error rate target) were not jointly considered in these works. In \cite{ref22}, the latter problem was jointly considered, given a predetermined detection probability on the primary nodes activity (i.e., providing conditional expressions). A recent approach dealing efficiently with power savings and minimization of consumed energy is energy harvesting \cite{procieee}. Powering mobile devices by harvested energy from ambient sources and/or external transmission activities renders wireless networks not only environmentally friendly but also self-sustaining. From the CR perspective, \cite{ref32,ref33} investigated cognitive systems enabled with energy harvesting equipment, but cooperative spectrum sensing or relayed transmission was not considered. 

In this paper, we focus on the EE optimization of CR cooperative systems, while an opportunistic strategy incorporating energy harvesting is introduced. According to the proposed strategy, secondary nodes switch between data transmission and energy harvesting depending on their sensing decision on the existence of primary nodes activity. More specifically, a dual-hop relaying system with multiple relays is adopted for the secondary system, where the end-to-end ($e2e$) communication is facilitated via a best relay selection policy. AF relaying protocol is implemented via semi-blind relays and fixed gains, a rather cost-efficient solution. Cooperative spectrum sensing is performed in a fixed sensing time duration, followed by a reporting of relay sensing measurements to SBS. Upon the aggregate decision at SBS, secondary nodes enter into either the energy harvesting phase (if primary transmission is detected) or the transmission phase (if primary transmission is not detected). In order to avoid unexpected co-channel interference to primary nodes (e.g., during reporting, transmitting and backhaul signaling from SBS to relays), the average interference temperature constraint is considered.

\subsection{Contributions}
The main contributions of this paper are summarized as follows:
\begin{itemize}
	\item A novel hybrid mode of operation is introduced for the CR cooperative system, where the secondary nodes opportunistically switch between data transmission and energy harvesting, upon the detection of primary activity. 
	\item A new exact closed-form expression for the detection probability of the considered configuration is presented, assuming that all the involved channels undergo independent and non-identically distributed (i.n.i.d.) Rayleigh fading conditions (i.e., different link distances among the nodes, appropriate for practical applications).
	\item A new exact formula for an important system performance metric is derived, namely, the outage probability, when the secondary system enters into the data transmission phase. The practical scenario of outdated channel state information (CSI), between the reporting and transmission phases, is also considered. Moreover, an exact formula for the average harvested power is presented, when the secondary system enters into the corresponding phase. 
	\item Based on the aforementioned performance analysis, the overall energy consumption at each secondary node is formulated. An optimization problem aiming at minimizing this energy is introduced, whereas a necessary and sufficient optimality condition is manifested. Finally, based on the enclosed analysis, the optimal sensing time is numerically evaluated given the required constraints.   
\end{itemize}

\subsection{Organization of the Paper}
In Section \ref{Section II} the signal model and the proposed opportunistic strategy for energy efficient cognitive relaying systems is presented. Subsequently, closed-form expressions for key performance measures of the secondary system are included in Section \ref{Section III}. Based on these expressions, the energy consumed from each secondary node is analyzed in Section \ref{Section IV}, followed by the introduction of an optimization problem on minimizing the overall energy consumption of the CR cooperative system. Indicative numerical results and useful discussions are presented in Section \ref{Section V}, while some concluding remarks are presented in Section \ref{Section VI}.  

\subsection{Notation}
Throughout this paper, the following notations are used: $\mathbb{E}[\cdot]$ stands for the expectation operator and $\text{Pr}[\cdot]$ returns probability. Also, $f_{X}(\cdot)$ and $F_{X}(\cdot)$ denote, respectively, probability density function (PDF) and cumulative distribution function (CDF) of random variable (RV) $X$. Furthermore, $\Gamma(\cdot,\cdot)$ represents the upper incomplete Gamma function \cite[Eq. (8.350.2)]{tables}, $\text{Ei}(\cdot)$ is the exponential integral \cite[Eq. (8.211.1)]{tables}, $J_{0}(\cdot)$ denotes the zeroth order Bessel function of the first kind \cite[Eq. (8.411)]{tables}, $I_{0}(\cdot)$ is the zeroth order modified Bessel function of the first kind \cite[Eq. (8.431)]{tables}, and $K_{n}(\cdot)$ is the $n$th order modified Bessel function of the second kind \cite[Eq. (8.446)]{tables}.

\section{System Model}
\label{Section II}

\begin{figure}[!t]
\centering
\includegraphics[keepaspectratio,width=2.8in]{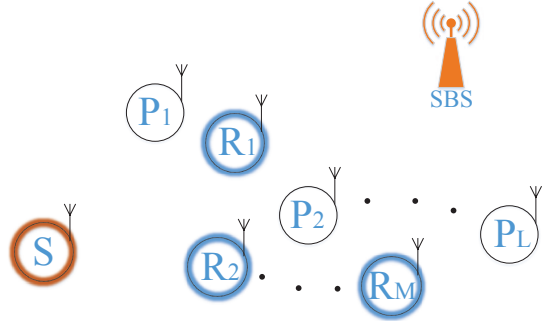}
\caption{The considered system configuration with a secondary source (S), $M$ relay (R) nodes and an SBS, i.e., a destination node (D), under the presence of $L$ primary (P) nodes.}
\label{fig1}
\end{figure}

\begin{figure}[!t]
\centering
\includegraphics[keepaspectratio,width=4.4in]{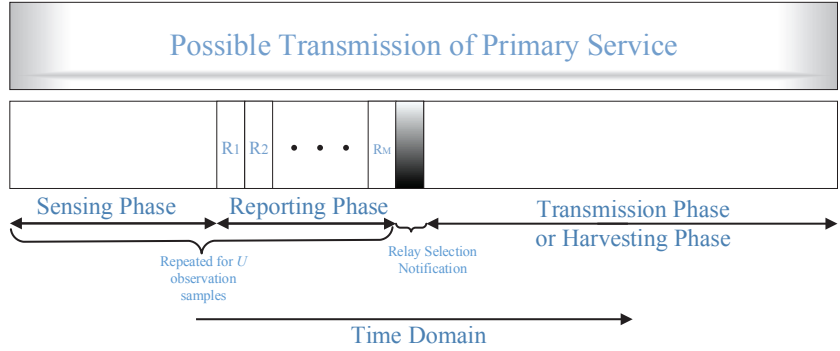}
\caption{Frame structure of the proposed approach.}
\label{fig2}
\end{figure}
A secondary (cognitive) dual-hop system consisted of a source ($S$) communicating with a destination ($D$) via $M$ relay ($R_{i}$ with $1\leq i\leq M$) nodes is considered.\footnote{Note that the terms $D$ and SBS will be interchangeably used in the rest of this paper.} Direct communication between source and destination is not available due to the long distance and strong propagation attenuation, while keeping in mind that in secondary systems the transmission power must in principle be maintained in quite low levels. The system operates in the vicinity of another licensed primary network, which consists of $L$ primary ($P_{j}$ with $1\leq j\leq L$) nodes (c.f., Fig. \ref{fig1}). In current study, we assume that all the involved signals are subject to i.n.i.d. Rayleigh fading as well as additive white Gaussian noise (AWGN) with a common power $N_{0}$. Thus, PDF of the instantaneous signal-to-noise ratio (SNR) is given by
\begin{align}
f_{\gamma_{i,j}}(x)=\frac{N_{0}\exp\left(-\frac{N_{0}x}{p_{i,j}\bar{\gamma}_{i,j}}\right)}{p_{i,j}\bar{\gamma}_{i,j}}, \ \ x\geq 0
\label{snrpdf}
\end{align}
where $\gamma_{i,j}$, $p_{i,j}$ and $\bar{\gamma}_{i,j}\triangleq d_{i,j}^{-\alpha_{i,j}}$ denote the instantaneous SNR, the signal power and the average received channel gain, respectively, from the $i$th to $j$th node. Moreover, $d_{i,j}$ and $\alpha_{i,j}$ represent the corresponding distance and path loss factor, respectively. Usually, $\alpha_{i,j} \in \{2,6\}$ denotes free-space loss to dense urban environmental conditions, correspondingly.

\subsection{Protocol Description}
The basic lines of the proposed approach are sketched in Fig. \ref{fig2}. The secondary nodes operate in a time division multiple access scheme, where sensing and transmission or harvesting phases are periodically alternating.

\subsubsection{Sensing Phase}
First, the relays and the destination enter into the \emph{sensing phase} where they listen to the presence of primary users' signals over the shared spectrum band within a fixed sensing duration $T_{S}$. We assume that the sensing time is smaller than the channel coherence time of primary transmission, such that the magnitude of the channel response remains constant within this phase. This is a reasonable assumption since the sensing duration is, in principle, maintained quite small (e.g., see \cite{ref16,ref17} and references therein). The received signal at the $i$th relay or destination can be expressed as
\begin{align}
y_{P,\mathcal{X}}=\sum^{L}_{j=1}\theta_{j}\sqrt{p_{p}}g_{P_{j},\mathcal{X}}s_{j}+n_{\mathcal{X}}, \ 1\leq i \leq M,\ \mathcal{X}\in\{R_{i},D\}
\label{recsignalsensing}
\end{align}
where $\theta_{j}=1$ or $0$ when the $j$th primary signal is present or absent, respectively. Also, $p_{p}$, $g_{P_{j},\mathcal{X}}$, $s_{j}$ and $n_{\mathcal{X}}$ denote the transmit power of primary nodes,\footnote{Without loss of generality and for the sake of clarity, a common power profile for the primary nodes is adopted.} the instantaneous channel gain from $P_{j}$ to $\mathcal{X}$, the transmitted data of the $j$th primary node and AWGN at $\mathcal{X}$, respectively. In particular, $\theta_{j}$ follows the Bernoulli probability mass function (PMF), such that $\theta_{j}=1$ or $\theta_{j}=0$ with probabilities $\text{Pr}[\theta_{j}=1]$ and $(1-\text{Pr}[\theta_{j}=1])$, respectively.

\subsubsection{Reporting Phase}
Next, each relay amplifies and forwards its local sensing measurement to the destination on its particular time slot (which is a priori reserved from the system), entering into the \emph{reporting phase}. Hence, the received signal at the destination, forwarded by the $i$th relay at its allocated time slot, yields as
\begin{align}
y_{R_{i},D}=\sqrt{p^{(\mathcal{R})}_{R_{i},D}}g_{R_{i},D}G_{\mathcal{R},i} y_{P,R_{i}}+n_{D}
\label{recsignaldestsensing}
\end{align}
where $p^{(\mathcal{R})}_{R_{i},D}$ and $g_{R_{i},D}$ are the transmission power and channel gain from $R_{i}$ to $D$, respectively, whereas $G_{\mathcal{R},i}$ denotes the fixed gain of the $i$th relay, all indicating the reporting phase.

In the sensing phase, each relay performs spectrum sensing on a per sample basis and then forwards its report to SBS. Thereby, defining that $U$ samples are provided within the sensing phase duration, the following procedure occurs: a) secondary relays sense the spectrum for sample-time $1$, then go to reporting phase for sample $1$; b) they sense the spectrum for sample-time $2$, then go to reporting phase for sample 2; and so on for $U$ samples. Notice that the latter alternation between sensing and reporting phases is a natural outcome of the AF relaying protocol's structure (receive and forward without further processing), whereas it further enhances the joint spatial-time diversity for the spectrum sensing process (since $U\geq 1$).    

\subsubsection{Harvesting Phase}
Then, the destination determines the presence of a primary transmission or not, according to the received signals' power. In fact, it compares the maximum from $M+1$ signals (from the relays and its own) with a predetermined power threshold value $\lambda$. In the case when this signal power is greater than $\lambda$, a detection event is declared in a subsequent time slot and all the relays initiate a harvesting phase, collecting energy from the occurring primary transmission(s). To this end, we have that
\begin{align}
P_{d}\triangleq \text{Pr}\left[\max_{u} \left\{\gamma_{P,D},\max_{i} \left\{\gamma^{(\mathcal{R})}_{e2e,i}\right\}^{M}_{i=1}\right\}^{U}_{u=1}\geq \lambda\right]
\label{maxsigdet}
\end{align}
where $P_{d}$ stands for the detection probability, while $\gamma^{(\mathcal{R})}_{e2e,i}$ represents the $e2e$ SNR at the sensing phase from the (potentially active) primary nodes to destination via the $i$th relay. For analytical tractability, we assume independence amongst the samples, which, nonetheless, represents a common yet efficient approach (e.g., see \cite{ref155555,ref1555555} and references therein). Moreover, the harvested power at the $i$th relay conditioned on a detection event, reads as \cite[Eq. (2)]{refharv}
\begin{align}
E_{H,i}=\eta P_{d}p_{p}\sum^{L}_{j=1}\theta_{j}\bar{\gamma}_{P_{j},R_{i}}\left|g_{P_{j},R_{i}}\right|^{2}
\label{energyharv}
\end{align}
where $\eta \in (0,1]$ is the radio frequency-to-direct current (RF-to-DC) conversion efficiency.

\subsubsection{Transmission Phase}
On the other hand, if the power of the signal in (\ref{maxsigdet}) is lower than $\lambda$, primary transmission is not detected with a probability $1-P_{d}$. In such a case, capitalizing on the status of channel gains from all the relay-to-destination links (collected from the aforementioned reporting phase), the destination selects the relay with the highest instantaneous SNR (i.e., $R_{s}$ with $s$ determined by the condition $\gamma_{R_{s}D}=\max_{l}\{\gamma_{R_{l},D}\}^{M}_{l=1}$) and broadcasts this information in the subsequent time slot. In turn, the selected relay informs the secondary source to enter into the transmission mode of operation.\footnote{The event of no signal returning from any relay back to the source, at this stage, can be interpreted as a triggering of harvesting phase for the secondary source.} Based on this call, the secondary system enters into the \emph{transmission phase}, while the source communicates with the destination via the selected relay (all the other ones stay idle\footnote{In current study, we focus on single-band (low-complex/low-cost) AF relay nodes. In addition, we assume that all active transmissions in the vicinity of the secondary system are realized by either the primary or the secondary nodes in that frequency band.}).

In fact, the classical half-duplex dual-hop AF relaying protocol is established at this stage, where the source-to-relay and relay-to-destination links occur in orthogonal transmission phases (e.g., in two consecutive time slots). Hence, the received signals at the relay and destination are, respectively, given by
\begin{align}
r_{S,R_{s}}=\sqrt{p^{(\mathcal{T})}_{S,R_{s}}}h_{S,R_{s}}z+n_{R_{s}}
\label{recsignaltrans}
\end{align}
and
\begin{align}
r_{R_{s},D}=\sqrt{p^{(\mathcal{T})}_{R_{s},D}}\hat{h}_{R_{s},D}G_{\mathcal{T},s} r_{S,R_{s}}+n_{D}
\label{recsignaldestrans}
\end{align}
where $z$, $h_{S,R_{s}}$, $p^{(\mathcal{T})}_{R_{s},D}$ and $G_{T,s}$ correspond to the source data, instantaneous channel gain from $S$ to $R_{s}$, transmission power from $S$ to $R_{s}$, and fixed gain of the $s$th relay during the transmission phase, respectively. Also, $\hat{h}_{R_{s},D}$ denotes the channel estimate of the selected relay to the destination, based on the instantaneous channel status derived from the previous reporting phase. It is noteworthy that $\hat{h}_{R_{s},D}$ could vary from the actual $h_{R_{s},D}$ due to a possible \emph{outdated} CSI at the destination. This condition is realized when a feedback delay and/or rapidly varying fading channels between the reporting and transmission phases are present. As such, the channel estimate is formed as \cite{ref1}
\begin{align}
\hat{h}_{R_{s},D}\triangleq \rho_{s} h_{R_{s},D}+\left(\sqrt{1-\rho^{2}_{s}}\right)w_{R_{s},D}
\label{channestimate}
\end{align} 
where $w_{R_{s},D}$ is a circularly symmetric complex Gaussian RV with the same variance as $h_{R_{s},D}$, while $\rho_{s}$ denotes the time correlation coefficient between $\hat{h}_{R_{s},D}$ and $h_{R_{s},D}$, defined as \cite[Eq. (4.1-63)]{ref2}
\begin{align}
\rho_{s}\triangleq J_{0}\left(2\pi f^{D}_{R_{s},D}T^{(s)}_{\text{diff}}\right)
\label{rho}
\end{align}
where $f^{D}_{i,j}$ and $T^{(s)}_{\text{diff}}$ are the maximum Doppler frequency on the $i-j$ link and the time difference between the actual instantaneous channel status and its corresponding estimate, respectively. It follows that in the case when the channel instances remain constant between the reporting and transmission phase, then $\rho=1$ and $\hat{h}_{R_{s},D}=h_{R_{s},D}$.

\subsection{Transmission Power of Secondary Nodes}
Although the transmission power of the primary service takes arbitrary values, this condition does not apply for the secondary system. We adopt an average interference constraint for the transmission power of secondary nodes, taking also into consideration the maximum output power, namely, $P_{\text{max}}$. Thereby, for the reporting phase, the following condition should be satisfied
\begin{align}
p^{(\mathcal{R})}_{R_{i},D}=\min\left\{P_{\text{max}},\frac{Q}{\mathbb{E}\left[q_{R_{i}}\right]}\right\},\ \ 1\leq i \leq M
\label{preport}
\end{align}
where $Q$ represents the received power threshold that should not be exceeded at the primary nodes, and 
\begin{align}
q_{R_{i}}\triangleq \max_{j}\left\{\gamma_{i,P_{j}}\right\}^{L}_{j=1}.
\label{q}
\end{align}

In the case when the system enters into the transmission phase, we have from \cite[Eq. (6)]{ref3} that
\begin{align}
p_{S,R_{s}}=\min\left\{P_{\text{max}},\frac{Q}{(1-P_{d})\mathbb{E}\left[q_{S}\right]}\right\}
\label{pstrans}
\end{align}
and
\begin{align}
p^{(\mathcal{T})}_{R_{s},D}=\min\left\{P_{\text{max}},\frac{Q}{(1-P_{d})\mathbb{E}\left[q_{R_{s}}\right]}\right\}
\label{ptrans}
\end{align}
where $q_{S}\triangleq \max_{j}\left\{\gamma_{S,P_{j}}\right\}^{L}_{j=1}$ and $q_{R_{s}}=\max\{\gamma_{R_{s},P_{j}}\}^{L}_{j=1}$.

\section{Performance Analysis}
\label{Section III}
We commence by describing the instantaneous $e2e$ received SNR at the destination in the aforementioned phases of the proposed strategy. Then, the detection probability is derived, followed by some important system performance measures for the secondary system, namely: i) the average harvested power of each relay during the harvesting phase, and ii) the $e2e$ outage probability during the transmission phase.

\subsection{SNR Statistics in the Sensing Phase}
Based on (\ref{recsignalsensing}), the received SNR from primary nodes to $D$ involves the sum of i.n.i.d. exponential RVs (with different link distances) and is obtained by\footnote{Since we model the signals as RVs with known transmission power, energy detector is adopted at the receiver, as being the optimal technique to detect the primary transmission(s) \cite{ref26}.}
\begin{align}
\gamma_{P,D}\triangleq \frac{p_{p}}{N_{0}}\sum^{L}_{j=1}\theta_{j}\gamma_{P_{j},D}.
\label{gpd}
\end{align}
According to the total probability theorem, PDF of $\gamma_{P,D}$ is derived as 
\begin{align}
f_{\gamma_{P,D}}(x)=\sum^{L}_{r=1}\text{Pr}\left[\left(\sum_{j}[\theta_{j}=1]\right)=r\right]f_{\gamma_{P,D}|r}(x),
\label{fgpd11}
\end{align}
where $\text{Pr}[(\sum_{j}[\theta_{j}=1])=r]$ denotes the probability that $r$ primary nodes are active (transmitting), given by
\begin{align*}
\text{Pr}\left[\left(\sum_{j}[\theta_{j}=1]\right)=r\right]\triangleq f_{r}=\binom{L}{r}\text{Pr}[\theta_{j}=1]^{r}(1-\text{Pr}[\theta_{j}=1])^{L-r},
\end{align*}
while $f_{\gamma_{P,D}|r}(x)$ denotes the conditional PDF of $\gamma_{P,D}$ given $r$ primary signal transmissions, which reads as \cite[Eq. (5)]{ref5}
\begin{align}
f_{\gamma_{P,D}|r}(x)=\left(\frac{N_{0}}{p_{p}}\right)\sum^{r}_{k=1}\left(\prod^{k}_{\begin{subarray}{c}j=1\\j\neq k\end{subarray}}\frac{1}{\bar{\gamma}_{P_{k},D}-\bar{\gamma}_{P_{j},D}}\right)\exp\left(-\frac{N_{0}x}{p_{p}\bar{\gamma}_{P_{k},D}}\right).
\label{fgpd1}
\end{align}
Therefore, the unconditional PDF of $\gamma_{P,D}$ is expressed as
\begin{align}
f_{\gamma_{P,D}}(x)=\left(\frac{N_{0}}{p_{p}}\right)\sum^{L}_{r=1}f_{r}\sum^{r}_{k=1}\left(\prod^{k}_{\begin{subarray}{c}j=1\\j\neq k\end{subarray}}\frac{1}{\bar{\gamma}_{P_{k},D}-\bar{\gamma}_{P_{j},D}}\right)\exp\left(-\frac{N_{0}x}{p_{p}\bar{\gamma}_{P_{k},D}}\right).
\label{fgpd}
\end{align}

Moreover, CDF of $\gamma_{P,D}$ is directly obtained as
\begin{align}
F_{\gamma_{P,D}}(x)=1-\sum^{L}_{r=1}f_{r}\sum^{r}_{k=1}\left(\prod^{k}_{\begin{subarray}{c}j=1\\j\neq k\end{subarray}}\frac{1}{\bar{\gamma}_{P_{k},D}-\bar{\gamma}_{P_{j},D}}\right)\bar{\gamma}_{P_{k},D}\exp\left(-\frac{N_{0}x}{p_{p}\bar{\gamma}_{P_{k},D}}\right).
\label{Fgpd}
\end{align}

\subsection{SNR Statistics in the Reporting Phase}
Based on (\ref{recsignalsensing}), the SNR of $P-R_{i}-D$ link (with $1\leq i\leq M$) is given by
\begin{align}
\gamma^{(\mathcal{R})}_{e2e,i}=\frac{\gamma^{(\mathcal{R})}_{1,i}\gamma^{(\mathcal{R})}_{2,i}}{\gamma^{(\mathcal{R})}_{2,i}+\mathcal{U}^{(\mathcal{R})}_{i}}
\label{ge2ereport}
\end{align}
where 
\begin{align}
\gamma^{(\mathcal{R})}_{1,i}\triangleq \frac{p_{p}}{N_{0}}\sum^{L}_{j=1}\theta_{j}\gamma^{(\mathcal{R})}_{P_{j},R_{i}}\text{ and }\gamma^{(\mathcal{R})}_{2,i}\triangleq \frac{p^{(\mathcal{R})}_{R_{i},D}}{N_{0}}\gamma^{(\mathcal{R})}_{R_{i},D}.
\label{g12report}
\end{align}
The parameter $\mathcal{U}^{(\mathcal{R})}_{i}=1/\left(G^{2}_{\mathcal{R},i}N_{0}\right)$ indicates a constant parameter, which is related to the value of fixed gain of the $i$th relay for the reporting phase. Among some popular precoding designs for this parameter, there is quite an efficient one \cite{ref4}, yielding 
\begin{align}
\mathcal{U}^{(\mathcal{R})}_{i}\triangleq \left(\mathbb{E}\left[\frac{1}{\gamma^{(\mathcal{R})}_{1,i}+1}\right]\right)^{-1}.
\label{U}
\end{align}

\begin{lem}
A closed-form expression for the CDF of $\gamma^{(\mathcal{R})}_{e2e,i}$ under i.n.i.d. Rayleigh fading channels reads as
\begin{align}
\nonumber
F_{\gamma^{(\mathcal{R})}_{e2e,i}}(x)&=1-\sum^{L}_{r=1}f_{r}\sum^{r}_{k=1}\frac{2 N_{0}\sqrt{\bar{\gamma}_{P_{k},R_{i}}\mathcal{U}^{(\mathcal{R})}_{i}x}}{\sqrt{p_{p}}(p^{(\mathcal{R})}_{R_{i},D}\bar{\gamma}_{R_{i},D})^{\frac{3}{2}}}\left(\prod^{k}_{\begin{subarray}{c}j=1\\j\neq k\end{subarray}}\frac{1}{\bar{\gamma}_{P_{k},D}-\bar{\gamma}_{P_{j},D}}\right)\\
&\times \exp\left(-\frac{N_{0}x}{p_{p}\bar{\gamma}_{P_{k},R_{i}}}\right) K_{1}\left(2N_{0}\sqrt{\frac{\mathcal{U}^{(\mathcal{R})}_{i}x}{p_{p}p^{(\mathcal{R})}_{R_{i},D}\bar{\gamma}_{P_{k},R_{i}}\bar{\gamma}_{R_{i},D}}}\right)
\label{fg12report}
\end{align}
\end{lem}
with
\begin{align}
\mathcal{U}^{(\mathcal{R})}_{i}=\left(\sum^{L}_{r=1}f_{r}\sum^{r}_{k=1}\left(\prod^{k}_{\begin{subarray}{c}j=1\\j\neq k\end{subarray}}\frac{1}{\bar{\gamma}_{P_{k},D}-\bar{\gamma}_{P_{j},D}}\right)\frac{\Gamma\left(0,\frac{N_{0}}{p_{p}\bar{\gamma}_{P_{k},R_{i}}}\right)}{\exp\left(-\frac{N_{0}}{p_{p}\bar{\gamma}_{P_{k},R_{i}}}\right)}\right)^{-1}
\label{Uclosedform}
\end{align}

\begin{align}
p^{(\mathcal{R})}_{R_{i},D}=\left(\frac{1}{P_{\text{max}}}+\frac{\mathbb{E}[q_{R_{i}}]}{Q}\right)^{-1}
\label{prdreport}
\end{align}
and
\begin{align}
\mathbb{E}[q_{R_{i}}]=\sum^{L}_{l=1}\sum^{L}_{k=0}\underbrace{\sum^{L}_{n_{1}=1}\cdots \sum^{L}_{n_{k}=1}}_{n_{1}\neq \cdots \neq n_{k}\neq k}\frac{(-1)^{k}}{k!\bar{\gamma}_{R_{i},P_{l}}\left(\frac{1}{\bar{\gamma}_{R_{i},P_{l}}}+\sum^{k}_{t=1}\frac{1}{\bar{\gamma}_{R_{i},P_{n_{t}}}}\right)^{2}}.
\label{Eq}
\end{align}

\begin{proof}
The proof is relegated in Appendix \ref{appa}.
\end{proof}

\subsection{Detection Probability}
Capitalizing on the aforementioned statistics, $P_{d}$ can be obtained in a closed-form expression. 
\begin{prop}
Detection probability for a given power threshold $\lambda$, $P_{d}(\lambda)$, is given by
\begin{align}
\nonumber
P_{d}(\lambda)&=1-\left[\prod^{M}_{i=1}F_{\gamma^{(\mathcal{R})}_{e2e,i}}(\lambda)F_{\gamma_{P,D}}(\lambda)\right]^{U}\\
&=1-\Delta(\lambda)^{U},
\label{pdclosed}
\end{align}
where 
\begin{align}
\nonumber
&\Delta(\lambda)=\\
\nonumber
&\prod^{M}_{i=1}\vast(1-\sum^{L}_{r=1}f_{r}\sum^{r}_{k=1}\frac{2 N_{0}\sqrt{\bar{\gamma}_{P_{k},R_{i}}\mathcal{U}^{(\mathcal{R})}_{i}\lambda}}{\sqrt{p_{p}}(p^{(\mathcal{R})}_{R_{i},D}\bar{\gamma}_{R_{i},D})^{\frac{3}{2}}}\left(\prod^{k}_{\begin{subarray}{c}j=1\\j\neq k\end{subarray}}\frac{1}{\bar{\gamma}_{P_{k},D}-\bar{\gamma}_{P_{j},D}}\right)\frac{K_{1}\left(2N_{0}\sqrt{\frac{\mathcal{U}^{(\mathcal{R})}_{i}\lambda}{p_{p}p^{(\mathcal{R})}_{R_{i},D}\bar{\gamma}_{P_{k},R_{i}}\bar{\gamma}_{R_{i},D}}}\right)}{\exp\left(\frac{N_{0}\lambda}{p_{p}\bar{\gamma}_{P_{k},R_{i}}}\right)}\vast)\\
&\times \left(1-\sum^{L}_{r=1}f_{r}\sum^{r}_{k=1}\left(\prod^{k}_{\begin{subarray}{c}j=1\\j\neq k\end{subarray}}\frac{1}{\bar{\gamma}_{P_{k},D}-\bar{\gamma}_{P_{j},D}}\right)\bar{\gamma}_{P_{k},D}\exp\left(-\frac{N_{0}\lambda}{p_{p}\bar{\gamma}_{P_{k},D}}\right)\right)
\label{delta}
\end{align}
and $U=T_{S}\mathcal{W}$ with $\mathcal{W}$ denoting the transmission bandwidth.
\end{prop}

\begin{proof}
Let $Y\triangleq \max\{y_{i}\}^{W}_{i=1}$. Then, the complementary CDF of $Y$, $\bar{F}_{Y}(x)\triangleq \text{Pr}[Y>x]=1-F_{Y}(x)$, becomes $\bar{F}_{Y}(x)=1-\prod^{W}_{i=1}F_{y_{i}}(x)$. Hence, first recall the channel coherence consistency within the duration of sensing time. Then, according to (\ref{maxsigdet}), while using (\ref{Fgpd}) and (\ref{fg12report}), (\ref{pdclosed}) is derived in a closed formulation. 
\end{proof}
At this point it should be mentioned that a similar approach (but with full-blind relays instead) was studied in \cite{ref26}, whereas assuming a single-sampled spectrum sensing, yielding a formula for the detection probability in terms of infinite series representation.

\subsection{Average Harvested Energy}
It is meaningful to investigate the average harvested energy of the $i$th relay ($1\leq i \leq M$), given a fixed harvesting phase duration (c.f., Fig. \ref{fig2}).

\begin{prop}
The average harvested power of the $i$th relay, defined as $\overline{E}_{H,i}$, which is collected upon a detection of the primary system's transmission, is given by
\begin{align}
\overline{E}_{H,i}=P_{d}(\lambda)\tilde{E}_{H,i},
\label{avgEi}
\end{align}
where
\begin{align}
\tilde{E}_{H,i}=\eta p_{p}\sum^{L}_{r=1}f_{r}\sum^{r}_{k=1}\left(\prod^{k}_{\begin{subarray}{c}j=1\\j\neq k\end{subarray}}\frac{1}{\bar{\gamma}_{P_{k},R_{i}}-\bar{\gamma}_{P_{j},R_{i}}}\right)\bar{\gamma}^{2}_{P_{k},R_{i}}.
\label{tildee}
\end{align}
while $P_{d}(\lambda)$ is given by (\ref{pdclosed}).\footnote{It is noteworthy that the corresponding average harvested power of the secondary source is directly obtained by substituting subscript $R_{i}$ with $S$ into (\ref{avgEi}), denoting the corresponding node.}  
\end{prop}

\begin{proof}
It holds that $\overline{E}_{H,i}\triangleq \mathbb{E}[E_{H,i}]=\int^{\infty}_{0}x f_{E_{H,i}}(x)dx$, while using (\ref{energyharv}) and (\ref{fgpd}), we have that
\begin{align*}
f_{E_{H,i}}(x)=\sum^{L}_{r=1}f_{r}\sum^{r}_{k=1}\left(\prod^{k}_{\begin{subarray}{c}j=1\\j\neq k\end{subarray}}\frac{1}{\left(\bar{\gamma}_{P_{k},R_{i}}-\bar{\gamma}_{P_{j},R_{i}}\right)}\right)\frac{\exp\left(-\frac{x}{\eta P_{d}p_{p}\bar{\gamma}_{P_{k},R_{i}}}\right)}{\eta P_{d}(\lambda)p_{p}}.
\end{align*} 
Then, after some simple manipulations, (\ref{avgEi}) is obtained.
\end{proof}

\subsection{Outage Probability in the Transmission Phase}
Following similar lines of reasoning as in the reporting phase, while based on (\ref{recsignaltrans}) and (\ref{recsignaldestrans}), the corresponding $e2e$ SNR of the $S-R_{s}-D$ link is given by
\begin{align}
\gamma^{(\mathcal{T})}_{e2e,s}=\frac{\gamma^{(\mathcal{T})}_{1,s}\gamma^{(\mathcal{T})}_{2,s}}{\gamma^{(\mathcal{T})}_{2,s}+\mathcal{U}^{(\mathcal{T})}_{s}}
\label{ge2etrans}
\end{align}
where 
\begin{align}
\gamma^{(\mathcal{T})}_{1,s}\triangleq \frac{p_{S,R_{s}}}{N_{0}}\left|h_{S,R_{s}}\right|^{2}\text{ and }\gamma^{(\mathcal{T})}_{2,s}\triangleq \frac{p^{(\mathcal{T})}_{R_{s},D}}{N_{0}}\left|\hat{h}_{R_{s},D}\right|^{2}.
\label{g12trans}
\end{align}
Also, $\mathcal{U}^{(\mathcal{T})}_{s}$ is explicitly defined in (\ref{U}), by substituting the superscript $(\cdot)^{\mathcal{R}}$ with $(\cdot)^{\mathcal{T}}$, denoting the transmission phase this time. 

In what follows, it is important to investigate the scenario of asymmetric channels for the secondary nodes (i.n.i.d. statistics), in a sense that relays usually keep unequal distances between the two secondary end nodes (i.e., $S$ and $D$) in practical dual-hop network setups. However, the scenario of symmetric channels (i.i.d. statistics), i.e., when the distances between the source-to-relay and relay-to-destination links are equal for each relay, is also considered, which can serve as a performance benchmark.

\subsubsection{Outage Probability}
Outage probability, $P_{\text{out}}$, is defined as the probability that the SNR of the $e2e$ $S-R_{s}-D$ link falls below a certain threshold value, $\gamma_{\text{th}}$, such that
\begin{align}
P_{\text{out}}(\gamma_{\text{th}})=\text{Pr}\left[\gamma^{(\mathcal{T})}_{e2e,s}\leq \gamma_{\text{th}}\right].
\label{poutdef}
\end{align}

\begin{lem}
CDF of the $e2e$ SNR for the $S-R_{s}-D$ link over Rayleigh fading channels is expressed as 
\begin{align}
F_{\gamma^{(\mathcal{T})}_{e2e,s}}(x)=1-2\:\Xi_{M}\exp\left(-\frac{N_{0}x}{p_{S,R_{s}}\bar{\gamma}_{S,R_{s}}}\right)\sqrt{\frac{N_{0}\mathcal{U}^{(\mathcal{T})}_{s}x}{Z_{M}p_{S,R_{s}}\bar{\gamma}_{S,R_{s}}}}K_{1}\left(2\sqrt{\frac{N_{0}\mathcal{U}^{(\mathcal{T})}_{s}Z_{M}x}{p_{S,R_{s}}\bar{\gamma}_{S,R_{s}}}}\right)
\label{cdftrans}
\end{align}
where
\begin{align*}
\Xi_{M}\triangleq \frac{\Psi_{M}}{(1-\rho^{2}_{l})p^{(\mathcal{T})}_{R_{l},D}\bar{\gamma}_{R_{l},D}\left(\Phi_{M}+\frac{\rho^{2}_{l}}{(1-\rho^{2}_{l})p^{(\mathcal{T})}_{R_{l},D}\bar{\gamma}_{R_{l},D}}\right)}
\end{align*}

\begin{align*}
Z_{M}\triangleq \frac{1}{(1-\rho^{2}_{l})p^{(\mathcal{T})}_{R_{l},D}\bar{\gamma}_{R_{l},D}}-\frac{\rho^{2}_{l}}{(1-\rho^{2}_{l})^{2}(p^{(\mathcal{T})}_{R_{l},D}\bar{\gamma}_{R_{l},D})^{2}\left(\Phi_{M}+\frac{\rho^{2}_{l}}{(1-\rho^{2}_{l})p^{(\mathcal{T})}_{R_{l},D}\bar{\gamma}_{R_{l},D}}\right)}
\end{align*}
with
\begin{align*}
&\Psi_{M}\triangleq \left\{
\begin{array}{c l}     
    \displaystyle \sum^{M}_{l=1}\sum^{M}_{k=0}\underbrace{\sum^{M}_{n_{1}=1}\cdots \sum^{M}_{n_{k}=1}}_{n_{1}\neq \cdots \neq n_{k}\neq k}\frac{(-1)^{k}}{k!p^{(\mathcal{T})}_{R_{l},D}\bar{\gamma}_{R_{l},D}}, & \text{for i.n.i.d.}\\
    & \\
    \displaystyle \sum^{M-1}_{l=0}\frac{\binom{M-1}{l}(-1)^{l}M}{p^{(\mathcal{T})}_{R,D}\bar{\gamma}_{R,D}}, & \text{for i.i.d.}
\end{array}\right.
\end{align*}

\begin{align*}
&\Phi_{M}\triangleq \left\{
\begin{array}{c l}     
    \displaystyle \frac{1}{p^{(\mathcal{T})}_{R_{l},D}\bar{\gamma}_{R_{l},D}}+\sum^{k}_{t=1}\frac{1}{p^{(\mathcal{T})}_{R_{n_{t}},D}\bar{\gamma}_{R_{n_{t}},D}}, & \text{for i.n.i.d.}\\
    & \\
    \displaystyle \frac{(l+1)}{p_{R,D}\bar{\gamma}^{(\mathcal{T})}_{R,D}}, & \text{for i.i.d.}
\end{array}\right.
\end{align*}

\begin{align}
p_{S,R_{s}}=\left(\frac{1}{P_{\text{max}}}+\frac{(1-P_{d})\mathbb{E}[q_{S}]}{Q}\right)^{-1}
\label{psrtrans}
\end{align}

\begin{align}
p^{(\mathcal{T})}_{R_{l},D}=\left(\frac{1}{P_{\text{max}}}+\frac{(1-P_{d})\mathbb{E}[q_{R_{l}}]}{Q}\right)^{-1}
\label{prdtrans}
\end{align}
and
\begin{align}
\mathcal{U}^{(\mathcal{T})}_{i}=\left(\left(\frac{N_{0}}{p_{S,R_{s}}\bar{\gamma}_{S,R_{s}}}\right)\frac{\Gamma\left(0,\frac{N_{0}}{p_{S,R_{s}}\bar{\gamma}_{S,R_{s}}}\right)}{\exp\left(-\frac{N_{0}}{p_{S,R_{s}}\bar{\gamma}_{S,R_{s}}}\right)}\right)^{-1}.
\label{utranss}
\end{align}
\end{lem}

\begin{proof}
The proof is provided in Appendix \ref{appb}.
\end{proof}

\begin{prop}
Outage probability of the $e2e$ SNR for the secondary system, during its transmission phase, is presented for i.n.i.d. and i.i.d. channels, respectively, as
\begin{align}
P^{(\text{i.n.i.d.})}_{\text{out}}(\gamma_{\text{th}})=\sum^{M}_{i=1}\text{Pr}[R_{i}]F_{\gamma^{(\mathcal{T})}_{e2e,i}}(\gamma_{\text{th}})
\label{pouttrans}
\end{align}
with
\begin{align}
\text{Pr}[R_{i}]\triangleq \left(1-\frac{\Xi_{M-1}}{\left(N_{0}+p_{R_{i},D}\bar{\gamma}_{R_{i},D}Z_{M-1}\right)Z_{M-1}}\right)
\label{prinid}
\end{align}
and
\begin{align}
P^{(\text{i.i.d.})}_{\text{out}}(\gamma_{\text{th}})=F_{\gamma^{(\mathcal{T})}_{e2e}}(\gamma_{\text{th}})
\label{pouttransiid}
\end{align}
where $F_{\gamma^{(\mathcal{T})}_{e2e,i}}(\cdot)$ is given by (\ref{cdftrans}), while $F_{\gamma^{(\mathcal{T})}_{e2e}}(\cdot)$ is obtained by dropping index-$i$ from the latter expression, reflecting to identical statistics.
\end{prop}

\begin{proof}
The proof is given in Appendix \ref{appc}.
\end{proof}

At this point, it should be stated that the previously derived expressions regarding the detection and outage probabilities along with the average harvested power are in a closed form including finite sum series, whereas they are exact. Hence, they are much more computationally efficient than existing methods so far (e.g., the rather demanding Monte Carlo (MC) simulations and/or manifold numerical integrations).

\subsection{Avoiding Saturation of Fixed Gain Relaying}
When the received signal at the first hop is slightly faded or the link-distance between the involved nodes is quite small, the aforementioned fixed gain may go into saturation \cite[\S V]{ref40}. To avoid such a scenario, we introduce the modified fixed gain, suitable for practical applications, used for both the reporting and transmitting phases correspondingly. Starting from the former phase and recalling (\ref{recsignaldestsensing}), it should hold for $G_{\mathcal{R},i}$ that
\begin{align}
G^{2}_{\mathcal{R},i}\left(p_{p}\sum^{L}_{j=1}\theta_{j}\left|g_{P_{j},R_{i}}\right|^{2}+N_{0}\right)\leq \mathcal{K}_{\mathcal{R}}
\label{modgainreq}
\end{align}
where $\mathcal{K}_{\mathcal{R}}$ represents a certain parameter that should not be exceeded during the AF process at the $i$th relay. A numerical solution of $\mathcal{K}_{\mathcal{R}}$ is provided in Appendix \ref{appd}. Thereby, substituting $G_{\mathcal{R},i}$ with $G_{\mathcal{R},i}\sqrt{\mathcal{K}_{\mathcal{R}}}$, the modified fixed gain is obtained, which results in clipping at the amplifier whenever (\ref{modgainreq}) is violated. Furthermore, it is straightforward to show that the modified gain for the transmission phase becomes $G_{\mathcal{T},i}\sqrt{\mathcal{K}_{\mathcal{T}}}$ with $\mathcal{K}_{\mathcal{T}}$ denoting the corresponding constant value used for this phase. The solution of $\mathcal{K}_{\mathcal{T}}$ can be obtained by following similar lines of reasoning as for the derivation of the aforementioned $\mathcal{K}_{\mathcal{R}}$.


\section{Energy Consumption}
\label{Section IV}
Motivated by the general interests towards green communications in emerging and future wireless systems, we deploy the results of the previous sections to analyze the average energy required per node in the proposed CR cooperative system. Subsequently, a necessary condition for minimizing this quantity is presented.

\subsection{Average Energy Consumption}
The average energy consumed at the $i$th relay in each frame (c.f., Fig. \ref{fig2}) is given by
\begin{align}
\overline{E}^{(i)}_{\text{total}}\triangleq \overline{E}_{S,i}T_{S}U+\overline{E}_{R,i}T_{R}U+(1-P_{d}(\lambda))\text{Pr}\left[R_{i}\right]\overline{E}_{T,i}T_{D}-\overline{E}_{H,i}T_{D}
\label{totale}
\end{align}
where $\overline{E}_{S,i}$, $\overline{E}_{R,i}$ and $\overline{E}_{T,i}$ are the average energy consumed at the sensing, reporting and transmission phase, respectively. In addition, $T_{S}$, $T_{R}$ and $T_{D}$ denote the duration of the sensing, reporting and transmission phase, correspondingly. Recall that $T_{S}\mathcal{W}=U$ stands for the time-bandwidth product, i.e., the number of samples required at the sensing phase. Notice from (\ref{totale}) that a transmission event occurs with a probability $1-P_{d}(\lambda)$, whereas $\text{Pr}\left[R_{i}\right]$ denotes the probability that the $i$th relay is selected for transmission, as provided in (\ref{prinid}) for i.n.i.d. and (\ref{priid}) for i.i.d. setups, respectively. On the other hand, as previously stated, all the relays enter into the harvesting phase with a probability $P_{d}(\lambda)$ during $T_{D}$, when a primary transmission is sensed. Notice that $P_{d}(\lambda)$ is already included within $\overline{E}_{H,i}$, by referring back to (\ref{avgEi}). 

The sensing energy can be considered identical for all the secondary nodes and, thus, (\ref{totale}) is simplified to
\begin{align}
\overline{E}^{(i)}_{\text{total}}\triangleq \overline{E}_{S}T_{S}U+\overline{E}_{R,i}T_{R}U+\left((1-P_{d}(\lambda))\text{Pr}\left[R_{i}\right]\overline{E}_{T,i}-P_{d}(\lambda)\tilde{E}_{H,i}\right)T_{D}.
\label{totale1}
\end{align}
Regarding $\overline{E}_{S}$, it holds that \cite{ref18} $\overline{E}_{S}=P_{R_{x}}$, where $P_{R_{x}}$ is the circuit power used to capture the received signal(s) power. Moreover, for $\overline{E}_{R,i}$, we have that \cite{ref18}-\cite{ref13} $\overline{E}_{R,i}=p^{(\mathcal{R})}_{R_{i},D}+P_{T_{x}}$, where $P_{T_{x}}$ is the circuit power used for signal transmission and $p^{(\mathcal{R})}_{R_{i},D}$ is given by (\ref{preport}). In general, both $P_{T_{x}}$ and $P_{R_{x}}$ are quite small, since each AF relay does not perform decoding, which is usually a more power-consuming operation \cite{ref18}. Similarly, $\overline{E}_{T,i}$ is obtained as \cite{ref12,ref13} $\overline{E}_{T,i}=p^{(\mathcal{T})}_{R_{i},D}+P_{T_{x}}$, where $p^{(\mathcal{T})}_{R_{i},D}$ is provided in (\ref{ptrans}). Finally, $\overline{E}_{H,i}$ is presented in (\ref{avgEi}) and, hence, the average energy consumption is obtained in a closed-form.

\subsection{Energy Consumption Minimization}
A fundamental requirement for the energy minimization problem is the provision of an appropriate time duration for the transmission of a certain amount of data bits. On the other hand, this duration should also satisfy an appropriate energy harvesting level (when primary transmission is detected), which minimizes the energy consumption. Based on the proposed mode of (hybrid) operation, while closely observing (\ref{totale1}), the tradeoff between sensing time and transmission/harvesting time duration is not trivial. 

First, we assume that $T_{R}$ is fixed, since only the transmission of a probe message is involved into the reporting phase. Then, a maximization of $T_{S}$ (i.e., a higher $U$) would reflect to a higher detection probability by capturing the transmission of primary nodes, at the cost of minimizing $T_{D}$. In turn, a shorter $T_{D}$ would reduce the amount of transmitted data (in the case of transmission phase), whereas it would also reduce the amount of harvested energy (in the case of harvesting phase). Thereby, maximizing $T_{S}$ results to a higher energy consumption. On the other hand, minimizing $T_{S}$ is also not a fruitful option, since the detection probability is reduced and, hence, a potential primary transmission may not be captured accurately. In such a case, a packet collision could occur during data transmission, while causing unexpected co-channel interference to the primary service.

Based on (\ref{delta}) and (\ref{tildee}), (\ref{totale1}) becomes
\begin{align}
\nonumber
\overline{E}^{(i)}_{\text{total}}&\triangleq \overline{E}_{S}T_{S}U+\overline{E}_{R,i}T_{R}U+\left(\Delta(\lambda)^{T_{S}\mathcal{W}}\text{Pr}\left[R_{i}\right]\overline{E}_{T,i}-\tilde{E}_{H,i}\left(1-\Delta(\lambda)^{T_{S}\mathcal{W}}\right)\right)T_{D}\\
&=\overline{E}_{S}T^{2}_{S}\mathcal{W}+\overline{E}_{R,i}T_{R}T_{S}\mathcal{W}+\left(\Delta(\lambda)^{T_{S}\mathcal{W}}\text{Pr}\left[R_{i}\right]\overline{E}_{T,i}-\tilde{E}_{H,i}\left(1-\Delta(\lambda)^{T_{S}\mathcal{W}}\right)\right)(T-T_{S})
\label{totale111}
\end{align}
where $T\triangleq T_{\text{total}}-T_{R}$ is the remaining duration of each frame, which is assumed as a fixed constant, whereas $T_{\text{total}}$ is the total frame duration.
 
Then, the considered energy consumption minimization problem is formulated as
\begin{align}
\nonumber
&\min_{T_{S}\geq 0} \overline{E}^{(i)}_{\text{total}}\\
&\text{s.t. }\mathcal{D}^{(i)}\geq \mathcal{D}^{\star}
\label{optprooo}
\end{align}
where $\mathcal{D}^{\star}$ denotes a \emph{target} amount of data bits that should be transmitted during the transmission phase and $\mathcal{D}^{(i)}$ stands for the amount of transmitted data of the $i$th relay within a given transmission duration, which is defined as
\begin{align}
\nonumber
\mathcal{D}^{(i)}&\triangleq (1-P_{d}(\lambda))\text{Pr}\left[R_{i}\right]R\:T_{D}=\Delta(\lambda)^{T_{S}\mathcal{W}}R (T-T_{S})\\
&=\Delta(\lambda)^{T_{S}\mathcal{W}}\text{Pr}\left[R_{i}\right]\mathcal{W}\log_{2}\left(1+\gamma^{(\mathcal{T})}_{e2e,i}\right)(T-T_{S})
\label{data}
\end{align}
with $R\triangleq \mathcal{W}\log_{2}(1+\gamma^{(\mathcal{T})}_{e2e,i})$ denoting the transmission rate (in bps), while $\gamma^{(\mathcal{T})}_{e2e,i}$ is explicitly defined back in (\ref{ge2etrans}).

Keeping in mind that $\text{ln}(\Delta(\lambda))\leq 0$ (since $\Delta(\lambda)\leq 1$), it readily follows from (\ref{totale111}) that
\begin{align}
\frac{\partial^{2}\overline{E}^{(i)}_{\text{total}}}{\partial^{2}T_{S}}=2E_{S}\mathcal{W}-\Delta(\lambda)^{T_{S}\mathcal{W}}(\tilde{E}_{H,i}+\overline{E}_{T,i}\text{Pr}\left[R_{i}\right])\mathcal{W}\text{ln}(\Delta(\lambda))(2-(T-T_{S})\mathcal{W}\text{ln}(\Delta(\lambda)))\geq 0
\end{align}
revealing the convexity of the objective function. In addition, to extract the hidden convexity of the included constraint of (\ref{optprooo}), the following transformation is applied
\begin{align}
\nonumber
&\Delta(\lambda)^{T_{S}\mathcal{W}}\text{Pr}\left[R_{i}\right]\mathcal{W}\log_{2}\left(1+\gamma^{(\mathcal{T})}_{e2e,i}\right)(T-T_{S})\geq \mathcal{D}^{\star}\\
&\Leftrightarrow \underbrace{2^{\frac{\Delta(\lambda)^{-T_{S}\mathcal{W}}\mathcal{D}^{\star}}{(T-T_{S})\mathcal{W}\text{Pr}\left[R_{i}\right]}}-1-\gamma^{(\mathcal{T})}_{e2e,i}}_{\triangleq \mathcal{D}'^{(i)}}\leq 0.
\end{align}
Hence, the original optimization problem in (\ref{optprooo}) can be reformulated in the standard form as
\begin{align}
\nonumber
&\min_{T_{S}\geq 0} \overline{E}^{(i)}_{\text{total}}\\
&\text{s.t. }\mathcal{D}'^{(i)}\leq 0.
\label{optpro}
\end{align}
Also, it holds that
\begin{align}
\nonumber
&\frac{\partial^{2}\mathcal{D}'^{(i)}}{\partial^{2}T_{S}}=2^{\frac{\Delta(\lambda)^{-T_{S}\mathcal{W}}\mathcal{D}^{\star}}{(T-T_{S})\mathcal{W}\text{Pr}\left[R_{i}\right]}}\vast\{\text{ln}^{2}(2)\left(\frac{\Delta(\lambda)^{-T_{S}\mathcal{W}}\mathcal{D}^{\star}}{\text{Pr}\left[R_{i}\right]\mathcal{W}(T-T_{S})^{2}}-\frac{\Delta(\lambda)^{-T_{S}\mathcal{W}}\mathcal{D}^{\star}\text{ln}(\Delta(\lambda))}{\text{Pr}\left[R_{i}\right](T-T_{S})}\right)^{2}\\
&+\text{ln}(2)\Delta(\lambda)^{-T_{S}\mathcal{W}}\left(\frac{2 \mathcal{D}^{\star}}{\text{Pr}\left[R_{i}\right]\mathcal{W}(T-T_{S})^{3}}-\frac{2 \mathcal{D}^{\star}\text{ln}(\Delta(\lambda))}{\text{Pr}\left[R_{i}\right](T-T_{S})^{2}}+\frac{\mathcal{D}^{\star}\mathcal{W}\text{ln}^{2}(\Delta(\lambda))}{\text{Pr}\left[R_{i}\right](T-T_{S})}\right)\vast\}\geq 0.
\end{align}
Thus, both the objective and constraint functions are convex with respect to $T_{S}$, which implies that the Karush-Kuhn-Tucker (KKT) conditions are necessary and sufficient for the optimal solution, whereas a unique minimum value exists \cite{ref14}. Motivated by the convexity of the problem, we introduce the following Lagrangian multiplier, termed $\mu$, into the equation:
\begin{align}
\mathcal{L}\triangleq \overline{E}^{(i)}_{\text{total}}-\mu\mathcal{D}'^{(i)}, \ \mu\geq 0
\label{langr}
\end{align}
and we set
\begin{align}
\frac{\partial \mathcal{L}}{\partial T_{S}}=0.
\label{langr1}
\end{align}
Solving this system, it follows that
\begin{align}
\nonumber
&\mu=\frac{2^{-\frac{\Delta(\lambda)^{-T_{S}\mathcal{W}}\mathcal{D}^{\star}}{(T-T_{S})\mathcal{W}\text{Pr}\left[R_{i}\right]}}\Delta(\lambda)^{T_{S}\mathcal{W}}\text{Pr}\left[R_{i}\right](T-T_{S})^{2}\mathcal{W}}{\mathcal{D}^{\star}\text{ln}(2)(1-(T-T_{S})\mathcal{W}\text{ln}(\Delta(\lambda)))}\\
&\times \left(\tilde{E}_{H,i}+2\overline{E}_{S}T_{S}\mathcal{W}+\overline{E}_{R,i}T_{R}\mathcal{W}-\Delta(\lambda)^{T_{S}\mathcal{W}}(\tilde{E}_{H,i}+\overline{E}_{T,i}\text{Pr}\left[R_{i}\right])(1-(T-T_{S})\mathcal{W}\text{ln}(\Delta(\lambda)))\right).
\end{align}
Since $\mu \geq 0$, the following necessary condition for optimally minimizing the average energy consumption yields as:
\begin{align}
\tilde{E}_{H,i}+\overline{E}_{T,i}\text{Pr}\left[R_{i}\right]\leq \frac{\Delta(\lambda)^{-T^{\star}_{S}\mathcal{W}}(\tilde{E}_{H,i}+2\overline{E}_{S}T^{\star}_{S}\mathcal{W}+\overline{E}_{R,i}T_{R}\mathcal{W})}{1-(T-T^{\star}_{S})\mathcal{W}\text{ln}(\Delta(\lambda))}
\label{m1large}
\end{align}
where $T^{\star}_{S}$ denotes the optimal sensing time.

Interestingly, the instantaneous data rate does not affect the latter condition for the efficiency on the energy consumption since (\ref{m1large}) is independent of $\gamma^{(\mathcal{T})}_{e2e,i}$. Obviously, the detection probability (obtained from (\ref{pdclosed}) in a closed-form) plays a key role to the overall performance of the secondary system and to EE. Also, notice that for a sufficiently high detection probability (i.e., very low miss probability), $\Delta(\lambda)\rightarrow 0^{+}$, which implies that the right-hand side (RHS) of (\ref{m1large}) takes extremely high values. In turn, for $\Delta(\lambda)\rightarrow 1^{-}$, RHS of (\ref{m1large}) maintains low values. Hence, (\ref{m1large}) is more likely to occur for an increased detection efficiency and vice versa. In other words, from a green communications viewpoint, the harvested energy is greatly enhanced as compared to the energy consumed for data transmission, as long as a robust and quite accurate detection scheme is preserved in the considered system.

Overall, it is important to mention that in practical applications certain licensed spectrum bands allocated for primary users may be idle quite often (the so-called \emph{spectrum hole} effect \cite{refsphole}). In this case, the secondary system would operate in consecutive time-frames, utilizing the previously proposed opportunistic relay scheduling scheme. Hence, it is highly possible that one or more secondary relays may reach a minimum energy level, defined as $\overline{E}^{(\min)}_{\text{total}}$, because the system enters more the transmission phase rather than the harvesting phase at this time period. The value of $\overline{E}^{(\min)}_{\text{total}}$ can be considered as an application-dependent fixed parameter and/or predetermined by the system manufacturer. As an illustrative example, consider the scenario when the system realizes a spectrum hole and the $j$th secondary relay has reached $\overline{E}^{(\min)}_{\text{total}}$. Then, upon the next data transmission phase, the secondary receiver utilizes (\ref{m1large}) for the $j$th relay. If this condition is satisfied, the $j$th relay takes place in the following relay selection process. On the other hand, if the latter condition is not satisfied, the secondary receiver excludes the $j$th relay for a potential selection of the data transmission phase. Thus, the analytical formulation of $\text{Pr}\left[R_{i}\right]$ (with $1\leq i\leq M$ and $i\neq j$) is provided back in (\ref{prinid}) by substituting $M$ with $M-1$ to denote the absence of the $j$th relay during the selection process.

\section{Numerical Results}
\label{Section V}
In this section, numerical results are presented and cross-compared with MC simulations to assess our theoretical findings. Each MC run was conducted over $10^{6}$ RV trials. There is a perfect match between all the analytical and the respective simulation results and, therefore, the accuracy of the presented analysis is verified. Henceforth, for notational simplicity and without loss of generality, we assume a common time correlation coefficient, defined as $\rho$. Moreover, $\eta=0.35$, $N_{0}=-131$dBm, $\mathcal{W}=1$MHz, the system carrier frequency is $2.5$GHz, while, as previously mentioned, the number of sensing samples is defined as $U=T_{S}\mathcal{W}$. Also, the path-loss exponent is assumed fixed as $\alpha=4$, corresponding to a classical macro-cell urban environment \cite[Table 2.2]{ref15}. In addition, $\text{Pr}[\theta_{j}=1]=\text{Pr}[\theta_{j}=0]=0.5$ is assumed, without loss of generality, which is in agreement with other research works, e.g., \cite{ref13,refnewbjornson}. All the included link distances are normalized with a reference distance equal to $1$km. In addition, for clarity reasons, we assume that the distance between the $l$th primary node and secondary source equals the distance between the $l$th primary node and $i$th relay and the distance between the $l$th primary node and destination, i.e., $d_{P_{l},S}=d_{P_{l},R_{i}}=d_{P_{l},D}\triangleq d_{P_{l}}$ $\forall \:\:l,i$. In what follows (owing to the non-identical statistics of the included nodes), we consider the following link-distance scenarios for the primary nodes; for $L=1$ let $d_{P_{1}} \in \mathbb{R}^{+}$ (in km), while for $L>1$ it is assumed that
\begin{align}
d_{P_{l+1}}\triangleq d_{P_{l}}+0.01 \ \ \forall l \in \{1,L-1\}.
\label{dp}
\end{align}
Following similar lines of reasoning, when the link distances per hop for the secondary nodes are non-identical, it is assumed that 
\begin{align}
\nonumber
&d_{S,R_{i+1}}\triangleq d_{S,R_{i}}+0.005,\\
&d_{R_{i+1},D}\triangleq d_{R_{i},D}+0.005 \ \ \forall i \in \{1,M-1\}.
\label{dr}
\end{align} 

Figure \ref{fig3} illustrates the detection probability for various distances between primary and secondary nodes. Its performance is worse for higher $\lambda$ threshold values and/or the existence of fewer primary nodes, as expected. This occurs due to the fact that when more primary nodes are placed in the vicinity of the secondary nodes, transmitting using a relatively high power, their active presence is more likely to be detected and vice versa. Scanning the open technical literature, a generally accepted target on the detection probability is found to be $P^{\star}_{d}(\lambda)=90\%$ (e.g., see \cite{ref1544}). Then, as an illustrative example, when primary nodes keep a minimum distance of 0.4km from secondary relays and/or destination, the latter target is satisfied for the given system settings. 
\begin{figure}[!t]
\centering
\includegraphics[trim=2cm 0.2cm 3cm 1cm, clip=true,totalheight=0.3\textheight]{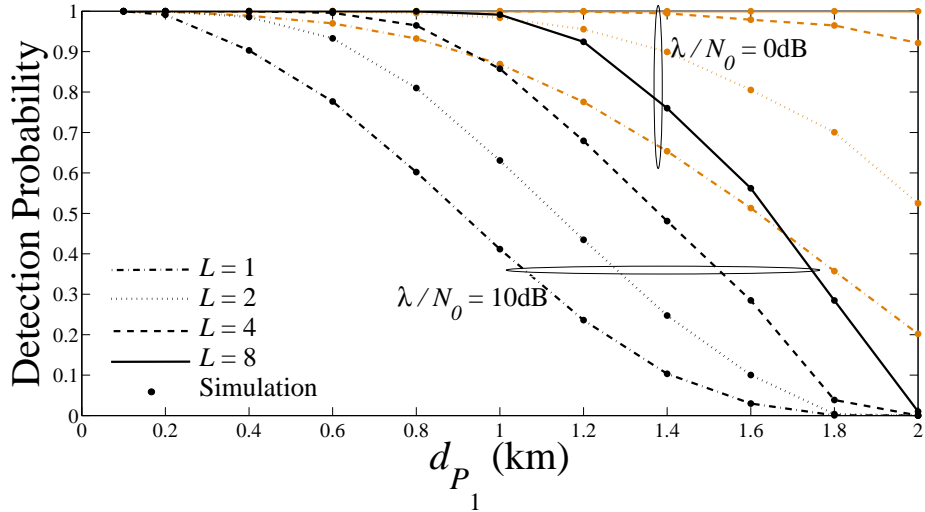}
\caption{$P_{d}(\lambda)$ vs. $d_{P_{1}}$. Recall (\ref{dp}), for the scenarios when $L>1$. Also, $M=1$, $U=200$, $d_{R,D}=0.1$km, $\frac{Q}{N_{0}}=2$dB, and $\{\frac{P_{\text{max}}}{N_{0}},\frac{p_{p}}{N_{0}}\}=10$dB.}
\label{fig3}
\end{figure}

In Fig. \ref{fig4}, outage performance of the secondary system is depicted for various $P_{\text{max}}$ values. It is worth noting that the diversity order is always one regardless of the number of relays, which is in agreement with \cite{ref19}. Moreover, when very rapidly varying fading channels are present (e.g., $\rho=0.1$), adding more relays does not alter the coding gain either (in fact, there is quite a marginal performance difference, which can be considered as negligible). Hence, the overall outage performance, i.e., both the coding and diversity gains, cannot be enhanced by maximizing $M$ in such environments. On the other hand, coding gain is improved by adding more relays into the system, when semi-constant channel fading conditions are realized (e.g., $\rho=0.9$). 

\begin{figure}[!t]
\centering
\includegraphics[trim=1.8cm 0.2cm 3cm 0cm, clip=true,totalheight=0.3\textheight]{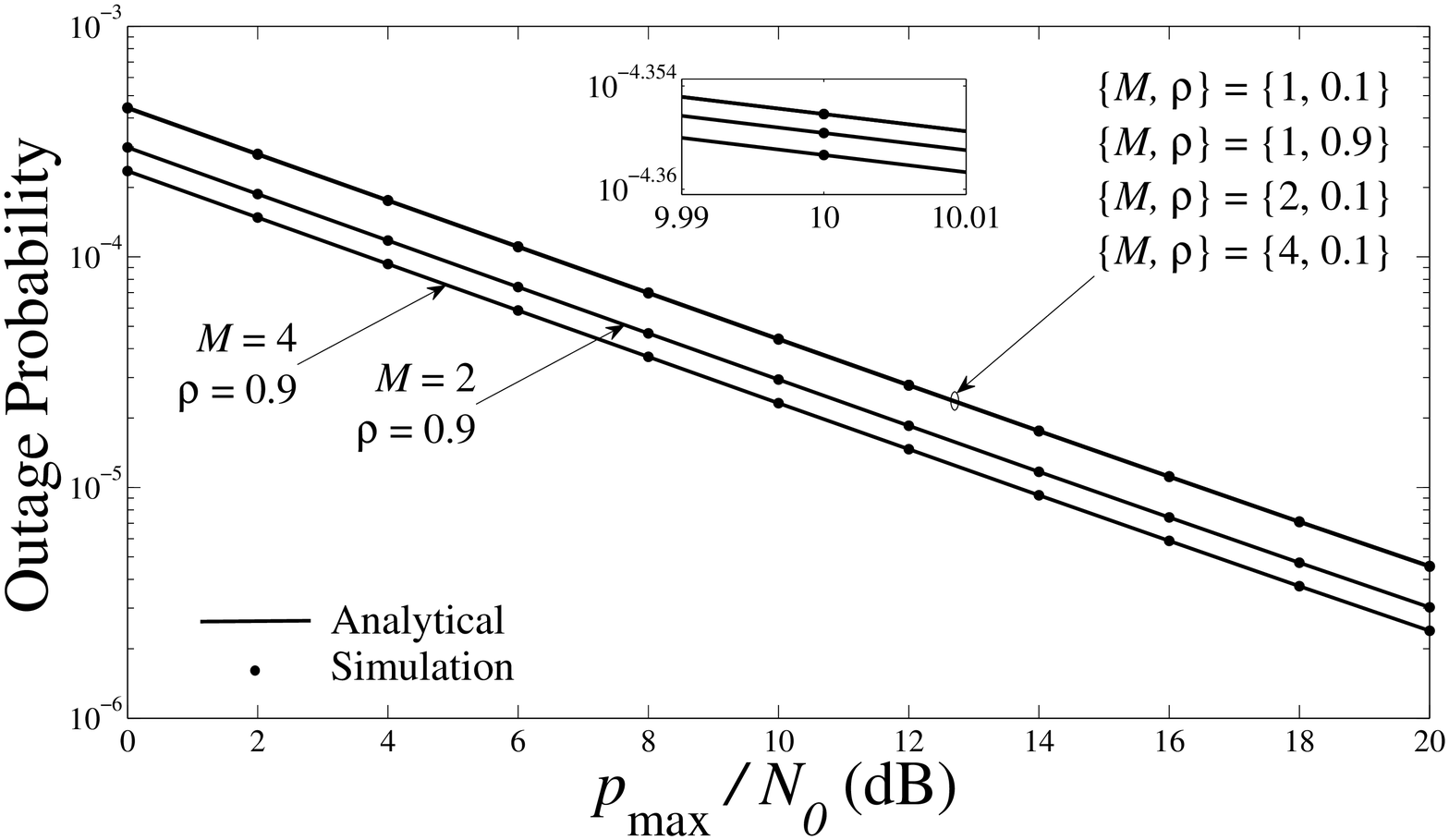}
\caption{$P_{\text{out}}(\gamma_{\text{th}})$ vs. $p_{\text{max}}/N_{0}$. Also, $L=2$, $d_{S,R}=d_{R,D}=0.1$km, $d_{P_{1},R}=0.3$km, $d_{P_{1},D}=0.4$km, $\{\frac{\lambda}{N_{0}},\frac{\gamma_{\text{th}}}{N_{0}}\}=3$dB, $\frac{Q}{N_{0}}=6$dB, and $\frac{p_{p}}{N_{0}}=30$dB.}
\label{fig4}
\end{figure}

Figures \ref{fig6} \ref{fig7}, and \ref{fig8} are devoted to the average energy consumed by using the proposed opportunistic strategy, from a green communications perspective. To this end, $\overline{E}_{\text{total}}$ is numerically evaluated using (\ref{totale1}) under different distances $d_{P_{1}}$ in Fig. \ref{fig6}. Obviously, there is an emphatic energy gain (on average) for each secondary node as this distance is relatively small (e.g., when $d_{P_{1}}<0.4$km). Also, EE is greatly enhanced for higher $L$ values (i.e., more primary nodes). This is a reasonable outcome since the overall average harvested energy is increased in such a scenario. In general, it can be seen that the average harvested energy is higher than the corresponding energy that is consumed for sensing, reporting and (potential) transmitting. This beneficial phenomenon stops holding only for the case when primary nodes are rather far-distant (e.g., when $d_{P_{1}}>1.1$km).

\begin{figure}[!t]
\centering
\includegraphics[trim=1.8cm 0.2cm 2.5cm 0cm, clip=true,totalheight=0.3\textheight]{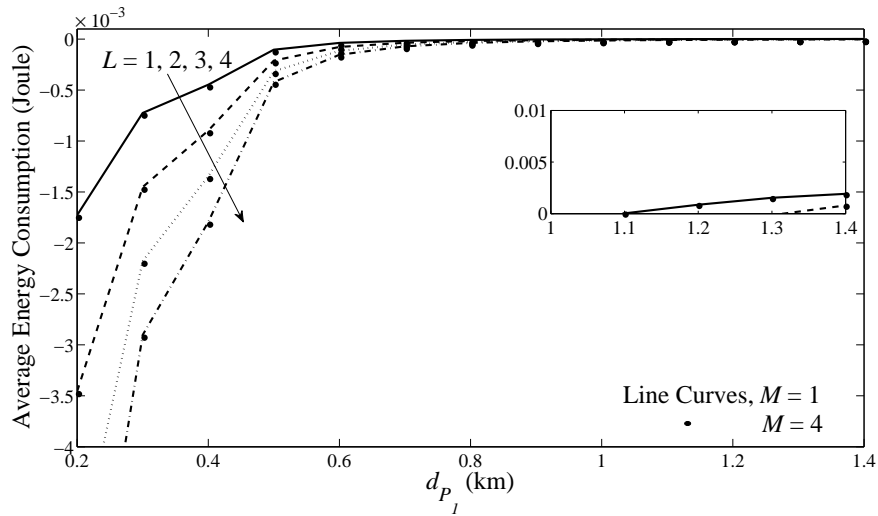}
\caption{$\overline{E}_{\text{total}}$ vs. various distances of the primary node-1 and the secondary system (i.e., $d_{P_{1}}$). Also, i.i.d. statistics between the $S-R$ and $R-D$ links are considered with $d_{R,D}=0.1$km and four relays (i.e., $M=4$). Other parameters are fixed as: $P_{\text{max}}=p_{p}=20$dBm, $\lambda=Q=17$dBm, $P_{T_{x}}=10$dBm, $P_{R_{x}}=9$dBm, $T_{\text{total}}=100$msec, $T_{R}=1$msec, and $T_{S}=20$msec (thus $T_{D}=79$msec).}
\label{fig6}
\end{figure}

\begin{figure}[!t]
\centering
\includegraphics[trim=0.2cm 0.2cm 0cm 0cm, clip=true,totalheight=0.3\textheight]{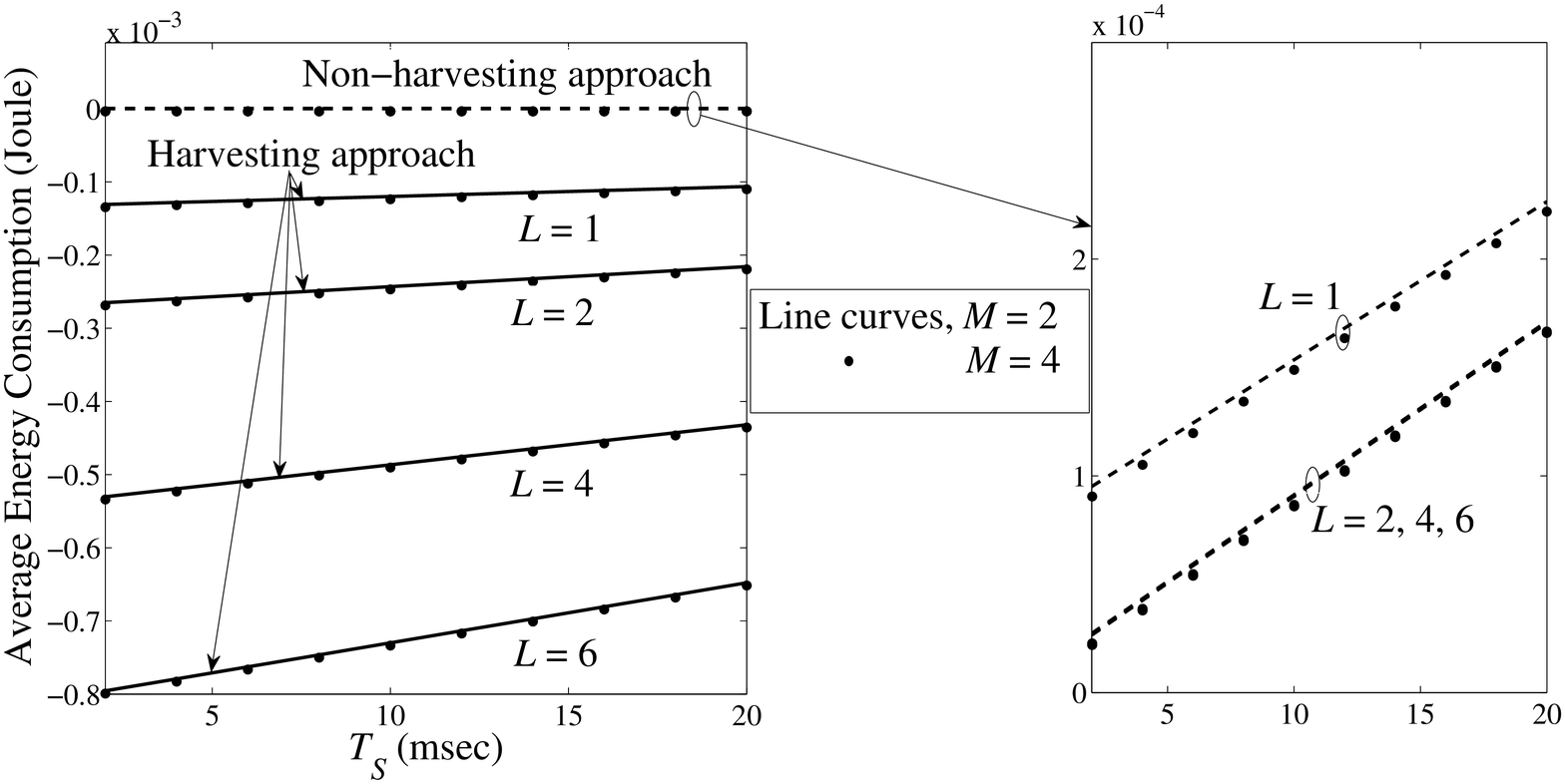}
\caption{$\overline{E}_{\text{total}}$ vs. various durations of the sensing time $T_{S}$. Also, $d_{P_{1}}=0.4$km. Other system parameters are the same as in Fig. \ref{fig6}.}
\label{fig7}
\end{figure}

\begin{figure}[!t]
\centering
\includegraphics[trim=1.8cm 0.2cm 2.5cm 0cm, clip=true,totalheight=0.3\textheight]{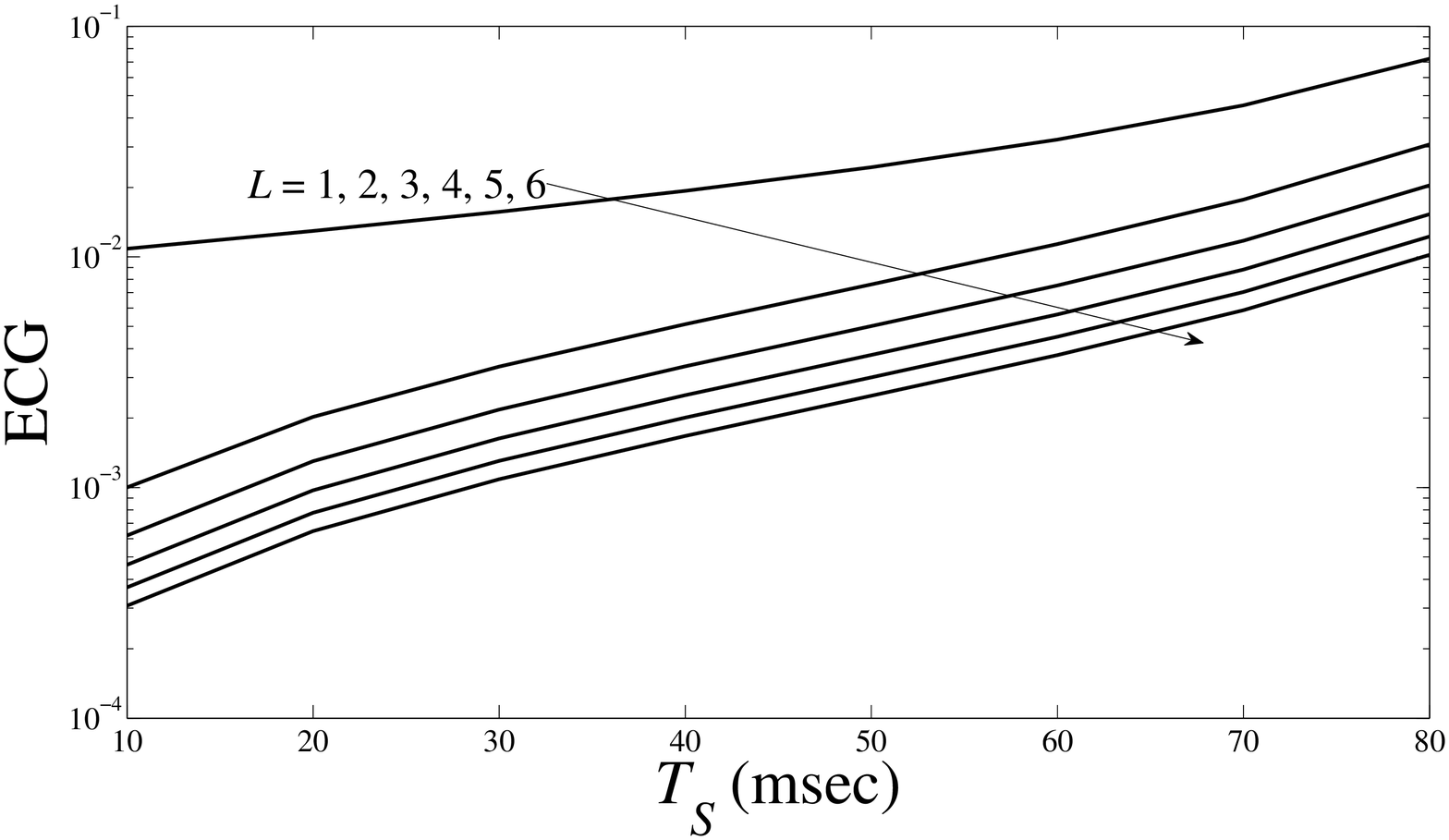}
\caption{ECG vs. various durations of the sensing time $T_{S}$. Also, $d_{P_{1}}=0.5$km, $d_{S,R_{1}}=d_{R_{1},D}=0.2$km, and $M=1$. Other system parameters are the same as in Fig. \ref{fig6}.}
\label{fig8}
\end{figure}

In Fig. \ref{fig7}, $\overline{E}_{\text{total}}$ is numerically evaluated for various durations of the sensing time $T_{S}$. Two different approaches are cross-compared; the proposed hybrid one (entitled as harvesting approach) and a suboptimal approach which does not support harvesting (entitled as non-harvesting approach). The latter approach, supports either secondary data transmission or an idle state for the included nodes in the case of primary detection. Hence, the lack of energy harvesting manifests itself since the right-most depiction of Fig. \ref{fig7} reveals a certain amount of consumed energy. On the other hand, the negative scale of $\overline{E}_{\text{total}}$ at the left-most depiction describes an overall EE and an opportunistically power-saving system setup.

An insightful observation obtained from Figs. \ref{fig6} and \ref{fig7} is the fact that the presence of more or less relays at the secondary system does not dramatically affect the average energy consumption. This occurs because when the system enters into the harvesting phase, upon the detection of primary transmission(s), all the included relays switch to the harvesting mode. Since the harvested energy is the dominant factor for EE and power savings, the derived results are quite straightforward. As previously noticed, performance gain from the introduction of more relays into the system appears only at the transmission stage of the secondary communication.

In Fig. \ref{fig8}, the energy consumed for sensing, reporting and transmitting data (given the transmission phase) is evaluated with respect to the harvested energy (given the harvesting phase). To this end, the energy-consumption gain is introduced, namely, ECG, which is defined, using (\ref{totale1}), as
\begin{align}
\text{ECG}\triangleq \frac{\overline{E}_{S}T_{S}+\overline{E}_{R,i}T_{R}T_{S}\mathcal{W}+(1-P_{d}(\lambda))\text{Pr}\left[R_{i}\right]\overline{E}_{T,i}T_{D}}{P_{d}(\lambda)\tilde{E}_{H,i}T_{D}}.
\end{align}
As can be seen from Fig. \ref{fig8}, the presence of two or more primary nodes provides an emphatic difference on power savings as compared to the rather overoptimistic scenario with only one primary node (i.e., when $L=1$). This occurs not only due to the maximization of power collected during the harvesting phase, but also due to quite an enhanced detection probability for $L\geq 2$.

Finally, Table \ref{Table} illustrates the optimal $T_{S}$ (which can be obtained with the aid of several numerical optimization tools\footnote{For instance, a global optimization algorithm that can be employed for numerical estimation of $T_{S}$ is the \emph{NMinimize} function within the MATHEMATICA$^{\text{TM}}$ software package.}), given the optimization problem in (\ref{optpro}) with its involved constraints. It can be readily seen that $T_{S}$ is greatly minimized for higher $L$ and $M$ values, especially when $M\geq 2$ and $L \geq 3$. This is a natural outcome gained from the spatial diversity of cooperative sensing by many relay nodes and, thus, from the rather efficient detection probability within a minimum sensing time.
\begin{table}[!t]
\caption{Optimal Sensing Time $T_{S}^{\star}$ (in sec) Obtained From (\ref{optpro})}
\label{Table}
\begin{center}
\begin{tabular}{l|| l l l l}\hline
\backslashbox{$M$}{$L$} & 1& 2& 3& 4\\\hline
 1 & 0.0873& 0.0815& 0.0658& 0.000115\\
 2 & 0.0873& 0.0815& 0.00372& 0.000115\\
 3 & 0.0873& 0.0815& 0.00372& 0.000115\\
 4 & 0.0873& 0.0815& 0.00370& 0.000106\\\hline
\end{tabular}
\end{center}
*$\{\frac{\lambda}{N_{0}},\frac{Q}{N_{0}}\}=7$dB, and $\{\frac{P_{\text{max}}}{N_{0}},\frac{p_{p}}{N_{0}}\}=30$dB, $d_{S,R_{1}}=d_{R_{1},D}=0.5$km, $d_{P_{1},R_{1}}=d_{P_{1},D}=1$km, $T=100$msec, $T_{R}=1$msec, and $R=100$Kbps.
\end{table}

\section{Conclusions}
\label{Section VI}
A new green cognitive dual-hop relaying system with multiple AF relays was proposed and its performance is analytically investigated. The novelty of this work relied on the introduction of a hybrid mode of operation for the involved secondary nodes, which opportunistically switch between data transmission and energy harvesting, correspondingly, upon the detection of primary network activity. New closed-form expressions regarding important system performance metrics were obtained over Rayleigh channel fading conditions, such as the detection probability, outage probability, and average harvested energy of each secondary node. Capitalizing on the derived expressions, the energy conservation of secondary nodes was further investigated, by modeling and analyzing the overall energy consumption minimization problem. A necessary and sufficient optimality condition for the proposed opportunistic strategy was presented, while some useful engineering insights were manifested. 

A natural extension of the proposed scheme could be the analytical performance evaluation and/or optimization when the secondary nodes utilize multi-band spectrum sensing/spectrum operation, which represents a challenging topic for future research.

\appendix
\subsection{Derivation of (\ref{fg12report}), (\ref{Uclosedform}), (\ref{prdreport}), and (\ref{Eq})}
\label{appa}
\numberwithin{equation}{subsection}
\setcounter{equation}{0}
CDF of the $e2e$ SNR for the $i$th secondary relay reads as \cite[Eq. (15)]{ref6}
\begin{align}
F^{(i)}_{\gamma^{(\mathcal{R})}_{e2e,i}}(x)=\int^{\infty}_{0}F_{\gamma^{(\mathcal{R})}_{1,i}}\left(x+\frac{\mathcal{U}^{(\mathcal{R})}x}{y}\right)f_{\gamma^{(\mathcal{R})}_{2,i}}(y)dy.
\label{Fge2ereport}
\end{align}
Based on (\ref{g12report}), $F_{\gamma^{(\mathcal{R})}_{1,i}}(x)$ can be directly obtained from (\ref{Fgpd}), by substituting $D$ with $R_{i}$. Additionally, $f_{\gamma^{(\mathcal{R})}_{2,i}}(x)$ stems as in (\ref{snrpdf}) by substituting $p_{i,j}$ and $\bar{\gamma}_{i,j}$ with $p^{(\mathcal{R})}_{R_{i},D}$ and $\bar{\gamma}_{R_{i},D}$, respectively.
Therefore, utilizing \cite[Eq. (3.471.9)]{tables} into (\ref{Fge2ereport}), (\ref{fg12report}) can be easily extracted.

Further, (\ref{U}) can be rewritten as
\begin{align}
\mathcal{U}^{(\mathcal{R})}_{i}=\left(\int^{\infty}_{0}(x+1)^{-1}f_{\gamma^{(\mathcal{R})}_{1,i}}(x)dx\right)^{-1}
\label{Uint}
\end{align}
while $f_{\gamma^{(\mathcal{R})}_{1,i}}(\cdot)$ is directly obtained from (\ref{fgpd}) by substituting subscript $D$ with $R_{i}$. Hence, utilizing \cite[Eq. (2.3.4.2)]{ref7} into (\ref{Uint}), (\ref{Uclosedform}) is obtained.

Regarding the derivation of (\ref{prdreport}), referring back to (\ref{preport}), we have that
\begin{align}
f_{\gamma^{(\mathcal{R})}_{2,i}}(x)=
\left\{
\begin{array}{c l}     
    \frac{N_{0}\exp\left(-\frac{N_{0}x}{P_{\text{max}}\bar{\gamma}_{R_{i},D}}\right)}{P_{\text{max}}\bar{\gamma}_{R_{i},D}}, & \mathbb{E}[q_{R_{i}}]<\frac{Q}{P_{\text{max}}},\\
    & \\
    \frac{N_{0}\mathbb{E}[q_{R_{i}}]\exp\left(-\frac{N_{0}\mathbb{E}[q_{R_{i}}]x}{Q\bar{\gamma}_{R_{i},D}}\right)}{Q\bar{\gamma}_{R_{i},D}}, & \mathbb{E}[q_{R_{i}}]>\frac{Q}{P_{\text{max}}}.
\end{array}\right.
\label{fg2reporttt}
\end{align}
Hence, it yields that
\begin{align}
\nonumber
F_{\gamma^{(\mathcal{R})}_{2,i}}(x)&=1-\left(1-F_{\gamma^{(\mathcal{R})}_{2,i}|P_{\text{max}}}(x)\right)\left(1-F_{\gamma^{(\mathcal{R})}_{2,i}|\frac{Q}{\mathbb{E}[q_{R_{i}}]}}(x)\right)\\
&=1-\exp\left(-\frac{N_{0}\left(\frac{1}{P_{\text{max}}}+\frac{\mathbb{E}[q_{R_{i}}]}{Q}\right)x}{\bar{\gamma}_{R_{i},D}}\right).
\label{Fg2reporttt}
\end{align}
By differentiating (\ref{Fg2reporttt}), the corresponding (unconditional) PDF of $\gamma^{(\mathcal{R})}_{2,i}$ is formed as in (\ref{snrpdf}) with the yielded transmission power $p^{(\mathcal{R})}_{R_{i},D}$ defined in (\ref{prdreport}).

Finally, since $\mathbb{E}[q_{R_{i}}]\triangleq \int^{\infty}_{0}xf_{q_{R_{i}}}(x)dx$, while based on (\ref{q}) and \cite[Eq. (15)]{ref8}, it holds that
\begin{align}
f_{q_{R_{i}}}(x)=\sum^{L}_{l=1}\sum^{L}_{k=0}\underbrace{\sum^{L}_{n_{1}=1}\cdots \sum^{L}_{n_{k}=1}}_{n_{1}\neq \cdots \neq n_{k}\neq k}\frac{(-1)^{k}}{k!\bar{\gamma}_{R_{i},P_{l}}}\exp\left(-\left(\frac{1}{\bar{\gamma}_{R_{i},P_{l}}}+\sum^{k}_{t=1}\frac{1}{\bar{\gamma}_{R_{i},P_{n_{t}}}}\right)x\right).
\label{fq}
\end{align}
Thus, after some simple algebra, (\ref{Eq}) arises. 

\subsection{Derivation of (\ref{cdftrans})}
\label{appb}
\numberwithin{equation}{subsection}
\setcounter{equation}{0}
First, the transmission power for the $S-R$ and $R-D$ links are determined, accordingly, based on (\ref{pstrans}) and (\ref{ptrans}). Similar to the derivation of (\ref{prdreport}), $p_{S,R_{s}}$ and $p^{(\mathcal{T})}_{R_{s},D}$ are obtained in (\ref{psrtrans}) and (\ref{prdtrans}), respectively. Moreover, the $S-R$ link follows a conditional (given the selected relay $R_{s}$) exponential distribution yielding
\begin{align}
F^{(\mathcal{T})}_{\gamma_{1,s}}(x)=1-\exp\left(-\frac{N_{0}x}{p_{S,R_{s}}\bar{\gamma}_{S,R_{s}}}\right).
\label{cdfg1trans}
\end{align}

Regarding the second-hop of the transmission phase, recall that $\gamma^{(\mathcal{T})}_{R_{s},D}$ and $\hat{\gamma}^{(\mathcal{T})}_{R_{s},D}$ have correlated exponential distributions with a corresponding conditional PDF given by \cite[Eq. (31)]{ref1}
\begin{align}
\nonumber
&f_{\gamma^{(\mathcal{T})}_{R_{s},D}|\hat{\gamma}^{(\mathcal{T})}_{R_{s},D}}(x|y)=\\
&\frac{\exp\left(-\frac{x+\rho^{2}_{s}y}{(1-\rho^{2}_{s})p_{R_{s},D}\gamma_{R_{s},D}}\right)}{(1-\rho^{2}_{s})p_{R_{s},D}\gamma_{R_{s},D}}I_{0}\left(\frac{2\rho_{s}\sqrt{x y}}{(1-\rho^{2}_{s})p_{R_{s},D}\gamma_{R_{s},D}}\right).
\label{joint}
\end{align}
From the theory of concomitants of ordered statistics \cite[\S 6.8]{reforderedstats}, it holds that
\begin{align}
f_{\gamma^{(\mathcal{T})}_{R_{s},D}}(x)=\int^{\infty}_{0}f_{\gamma^{(\mathcal{T})}_{R_{s},D}|\hat{\gamma}^{(\mathcal{T})}_{R_{s},D}}(x|y)f_{\hat{\gamma}^{(\mathcal{T})}_{R_{s},D}}(y)dy.
\label{uncond}
\end{align} 
Thus, since $R_{s}$ is selected based on maximizing the SNR of the relay-to-destination link, $f_{\hat{\gamma}^{(\mathcal{T})}_{R_{s},D}}(\cdot)$ becomes
\begin{align}
f_{\hat{\gamma}^{(\mathcal{T})}_{R_{s},D}}(y)=\Psi_{M}\exp(-\Phi_{M}y),
\label{fmax}
\end{align} 
by following similar lines of reasoning as for the derivation of (\ref{fq}). Recall that $\Psi_{M}$ and $\Phi_{M}$ are defined in (\ref{cdftrans}). Substituting (\ref{joint}) and (\ref{fmax}) into (\ref{uncond}), while utilizing \cite[Eq. (2.15.5.4)]{ref10}, it holds that
\begin{align}
f_{\gamma^{(\mathcal{T})}_{R_{s},D}}(x)=\Xi_{M}\exp(-Z_{M}x).
\label{uncond1}
\end{align} 
Further, using the integral in (\ref{Fge2ereport}), but plugging (\ref{cdfg1trans}) and (\ref{uncond1}), (\ref{cdftrans}) can be extracted. Finally, 
\begin{align*}
\mathcal{U}^{(\mathcal{T})}_{i}=\int^{\infty}_{0}(x+1)^{-1}f^{(\mathcal{T})}_{\gamma_{1,s}}(x)dx
\end{align*}
with $f^{(\mathcal{T})}_{\gamma_{1,s}}(\cdot)$ being the exponential PDF, yielding (\ref{utranss}).

\subsection{Derivation of (\ref{pouttrans}) and (\ref{pouttransiid})}
\label{appc}
\numberwithin{equation}{subsection}
\setcounter{equation}{0}
The unconditional outage probability of the received $e2e$ SNR in (\ref{pouttrans}) can be obtained by averaging  (\ref{cdftrans}) over all the relay selection probabilities, yielding
\begin{align}
P^{(\text{i.n.i.d.})}_{\text{out}}(\gamma_{\text{th}})=\sum^{M}_{i=1}\text{Pr}\left[R_{i}\right]F_{\gamma^{(\mathcal{T})}_{e2e,i}}(\gamma_{\text{th}}),
\label{pouttrans111}
\end{align}
where $\text{Pr}\left[R_{i}\right]$ is defined in (\ref{totale}). Suppose that $R_{i}$ is selected, which means that $\gamma_{R_{i},D}$ is the largest SNR among all $R-D$ links. Let $R_{v}=\max\{\gamma_{R_{j},D}\}^{M-1}_{j=1}$, except the selected one (i.e., $\gamma_{R_{i},D}> \gamma_{R_{v},D}$). Then, $\text{Pr}\left[R_{i}\right]$ can be written as
\begin{align}
\nonumber
\text{Pr}\left[R_{i}\right]&\triangleq \text{Pr}\left[\max\{\gamma_{R_{1},D},\ldots,\gamma_{R_{M-1},D}\}\leq \gamma_{R_{v},D}|\gamma_{R_{i},D}\right]\\
&=\int^{\infty}_{0}F_{\gamma_{R_{v},D}}(x)f_{\gamma_{R_{i},D}}(x)dx
\label{pri}
\end{align} 
where $F_{\gamma_{R_{v},D}}(\cdot)$ is the CDF of the maximum of $M-1$ i.n.i.d. exponential RVs. Based on (\ref{uncond1}), but substituting $M$ with $M-1$, it holds that
\begin{align}
F_{\gamma_{\gamma_{R_{v},D}}}(x)=1-\frac{\Xi_{M-1}}{Z_{M-1}}\exp(-Z_{M-1}x).
\label{uncond1cdf}
\end{align}  
Hence, substituting (\ref{uncond1cdf}) and (\ref{snrpdf}) into (\ref{pri}), (\ref{pouttrans}) is directly obtained.

In the special case of i.i.d. Rayleigh fading channels for the secondary system, we drop index-$i$ to ensure identical statistics per hop. To this end, the selection of each relay is equiprobable, yielding
\begin{align}
\text{Pr}\left[R_{i}\right]=\frac{1}{M}.
\label{priid}
\end{align}
Therefore, (\ref{pouttrans}) becomes (\ref{pouttransiid}), thus, completing the proof. 

\subsection{Numerical Solution of $\mathcal{K}_{\mathcal{R}}$}
\label{appd}
\numberwithin{equation}{subsection}
\setcounter{equation}{0}
When the channel gain conditions of the first hop result to a violation of (\ref{modgainreq}), the amplifier is clipped to $\mathcal{K}_{\mathcal{R}}$. Consequently, similar to \cite[Eq. (19)]{ref40}, the modified fixed gain of the $i$th relay is modeled as
\begin{align}
&\tilde{G}^{2}_{\mathcal{R},i} \triangleq \left\{
\begin{array}{c l}     
    \left(\mathcal{U}^{(\mathcal{R})}_{i}\right)^{-1}, & \gamma^{(\mathcal{R})}_{{1},i}\leq \mathcal{T}_{i}\\
    & \\
    \frac{1}{\gamma^{(\mathcal{R})}_{1,i}+1}, & \text{otherwise}
\end{array}\right.
\label{gmod}
\end{align} 
where $\mathcal{T}_{i}\triangleq \frac{\mathcal{K}_{\mathcal{R}}\mathcal{U}^{(\mathcal{R})}_{i}}{N_{0}}-1$ is a threshold value in terms of the received SNR at the $i$th relay, given directly from (\ref{modgainreq}), satisfying that $G^{2}_{\mathcal{R},i}(p_{p}\sum^{L}_{j=1}\theta_{j}|g_{P_{j},R_{i}}|^{2}+N_{0})= \mathcal{K}_{\mathcal{R}}$. Then, averaging (\ref{gmod}), the (unconditional) modified fixed gain yields as
\begin{align}
\nonumber
&\tilde{G}^{2}_{\mathcal{R},i}=\left(\mathcal{U}^{(\mathcal{R})}_{i}\right)^{-1}F^{(\mathcal{R})}_{\gamma_{1},i}(\mathcal{T})+\int^{\infty}_{\mathcal{T}}\left(x+1\right)^{-1}f^{(\mathcal{R})}_{\gamma_{1},i}(x)dx\\
&=\frac{F^{(\mathcal{R})}_{\gamma_{1},i}\left(\frac{\mathcal{K}_{\mathcal{R}}\mathcal{U}^{(\mathcal{R})}_{i}}{N_{0}}-1\right)}{\mathcal{U}^{(\mathcal{R})}_{i}}-\sum^{L}_{r=1}f_{r}\sum^{r}_{k=1}\left(\prod^{k}_{\begin{subarray}{c}j=1\\j\neq k\end{subarray}}\frac{1}{\bar{\gamma}_{P_{k},R_{i}}-\bar{\gamma}_{P_{j},R_{i}}}\right) \exp\left(\frac{N_{0}}{p_{p}\bar{\gamma}_{P_{k},R_{i}}}\right)\text{Ei}\left(-\frac{\mathcal{K}_{\mathcal{R}}\mathcal{U}^{(\mathcal{R})}_{i}}{p_{p}\bar{\gamma}_{P_{k},R_{i}}}\right)
\label{avggainfinal}
\end{align} 
where the second equality arises by utilizing \cite[Eq. (3.352.2)]{tables}.

Finally, since $G^{2}_{\mathcal{R},i}=1/(\mathcal{U}^{(\mathcal{R})}_{i}N_{0})$, $\mathcal{K}_{\mathcal{R}}$ can be numerically calculated by equating (\ref{avggainfinal}) with the reciprocal of (\ref{Uclosedform}), resulting to the following transcendental expression
\begin{align}
\nonumber
&\sum^{L}_{r=1}f_{r}\sum^{r}_{k=1}\left(\prod^{k}_{\begin{subarray}{c}j=1\\j\neq k\end{subarray}}\frac{1}{\bar{\gamma}_{P_{k},D}-\bar{\gamma}_{P_{j},D}}\right)\bar{\gamma}_{P_{k},D}\exp\left(-\frac{N_{0}\left(\frac{\mathcal{K}_{\mathcal{R}}\mathcal{U}^{(\mathcal{R})}_{i}}{N_{0}}-1\right)}{p_{p}\bar{\gamma}_{P_{k},D}}\right)\\
&+\mathcal{U}^{(\mathcal{R})}_{i}\sum^{L}_{r=1}f_{r}\sum^{r}_{k=1}\left(\prod^{k}_{\begin{subarray}{c}j=1\\j\neq k\end{subarray}}\frac{1}{\bar{\gamma}_{P_{k},R_{i}}-\bar{\gamma}_{P_{j},R_{i}}}\right)\exp\left(\frac{N_{0}}{p_{p}\bar{\gamma}_{P_{k},R_{i}}}\right)\text{Ei}\left(-\frac{\mathcal{K}_{\mathcal{R}}\mathcal{U}^{(\mathcal{R})}_{i}}{p_{p}\bar{\gamma}_{P_{k},R_{i}}}\right)-1+\frac{1}{N_{0}}=0
\label{numericalkr}
\end{align}
satisfying a fixed gain that avoids saturation of the amplifier.

\ifCLASSOPTIONcaptionsoff
  \newpage
\fi

\end{document}